\documentclass[11pt]{article}
\usepackage{amsmath,amsfonts,amsthm,amssymb,verbatim}
\usepackage{epsfig}

\numberwithin{equation}{section}

\newcommand{\e}{\varepsilon}
\newcommand{\rd}{{\rm d}}
\newcommand{\x}{{\rm x}}
\newcommand{\wh}{\widehat}
\newcommand{\ann}{a^{\hphantom{+}}}
\newcommand{\cre}{a^{\dagger}}

\newcommand{\beq}{\begin{equation}}
\newcommand{\eeq}{\end{equation}}

\newcommand{\beqa}{\begin{eqnarray}}
\newcommand{\eeqa}{\end{eqnarray}}
\newcommand\Z{{\mathbb Z}}
\newcommand\N{{\mathbb N}}

\newcommand\R{{\mathbb R}}
\newcommand\eps{\varepsilon}
\newcommand\half{\mbox{$\frac 12$}}

\newcommand\rhov\varrho

\newcommand\rmor{ \,\,\,{\rm or}\,\,\, }
\newcommand\rmand{ \,\,\,{\rm and}\,\,\, }

\newcommand\rhop{\rho_{\,\eps}}
\newcommand\rhom{\rho_{\,-\eps}}
\newcommand\const{{\rm const.\, }}

\newcommand\al{{\alpha}}
\newcommand\umt{\left\{u_1, m_1\right\},\cdots,\left\{u_t,m_t\right\}}

\newcommand\be{\beta}

\newcommand\mcut{m_c}

\renewcommand\kappa\varkappa
\renewcommand\rho\varrho

\newtheorem{thm}{THEOREM}[section]
\newtheorem{prop}{PROPOSITION}[section]
\newtheorem{lem}{Lemma}[section]
\newtheorem{cor}{COROLLARY}[section]
\newtheorem{mydef}{DEFINITION}[section]
\theoremstyle{definition}

\begin{document}


\title{The Second Order Upper Bound for the Ground Energy of a Bose Gas \thanks{Partially supported
by NSF grants DMS-0757425, 0804279}
}

\author{Horng-Tzer  Yau  and  Jun  Yin \\ \\
Department of Mathematics, Harvard University\\
Cambridge MA 02138, USA \\ \\
\\}





\maketitle

\begin{abstract} 
Consider $N$  bosons in a finite box $\Lambda= [0,L]^3\subset \mathbf R^3$
interacting via a two-body  smooth repulsive  short range potential. 
We construct a variational state which gives the following 
upper bound on the ground state energy
per particle 
\[
\overline\lim_{\rho\to0} \overline \lim_{L \to \infty, \, N/L^3 \to \rho} \left(
\frac{e_0(\rho)- 4 \pi a \rho }{(4 \pi a)^{5/2}( \rho )^{3/2} }\right)\leq \frac{16}{15\pi^2},
\]
where $a$ is the scattering length of the potential.  Previously, an upper bound 
of the form $C 16/15\pi^2$ for some constant $C > 1$ was obtained in \cite{ESY}.
Our result  proves the upper bound of the  prediction by Lee-Yang \cite{LYang} and  Lee-Huang-Yang \cite{LHY}.

\end{abstract} 

\renewcommand{\thefootnote}{${\,}$}
\renewcommand{\thefootnote}{${\, }$}
\footnotetext{\copyright\,2009 by the author.
This paper may be reproduced, in its entirety, for non-commercial
purposes.}

\bigskip


\bigskip



%

\newpage

\section{Introduction}

The ground state energy  is a fundamental property of a quantum system and 
it has been intensively studied since the invention of the quantum mechanics. 
The recent progresses  in experiments for the 
 Bose-Einstein condensation have inspired re-examination of the theoretic foundation  
 concerning the Bose system and, in particular, its ground state energy.    
 In the low density limit, the leading term of the ground state energy per particle 
 was identified  rigorously by   Dyson (upper
bound) \cite{D} and Lieb-Yngvason (lower bound) \cite{LY} to be $4\pi a \varrho$, 
where $a$ is the scattering length of the two-body
potential and $\varrho$ is the density. The famous second order
correction to this leading term was first computed by Lee-Yang \cite{LYang} (see also Lee-Huang-Yang \cite{LHY} and  the recent paper 
by Yang \cite{Y} for results in other dimensions). To describe this prediction, 
we now fix our notations:  Consider $N$ interacting bosons in a finite box $\Lambda= [0,L]^3\subset \mathbf R^3$
with periodic boundary conditions.  
The two-body interaction  is given by a smooth nonnegative potential $V$ 
of fast decay. 
The Lee-Yang's prediction of the energy per particle 
up to the second order  is given by
\beq\label{nextorder}
e_0(\rho)  = 4\pi\varrho a \Big [1+
\frac{128}{15\sqrt{\pi}}( \varrho a^3)^{1/2}
 + \cdots  \Big ] \; .
\eeq
The approach by Lee-Yang \cite{LYang} is based on the pseudo-potential approximation \cite{HY, LHY}
and the ``binary collision expansion method'' \cite{LHY}. One can also obtain \eqref{nextorder} by performing the 
Bogoliubov \cite{B} approximation and then replacing the integral of the potential 
by its scattering length \cite{Landau}.  Another derivation of \eqref{nextorder} 
was later given  by Lieb \cite{L} using 
a self-consistent closure assumption  for the hierarchy of correlation functions.

In the recent paper \cite{ESY},  the potential $V$ was replaced by $\lambda V_0$ for some fixed function $V_0$ and 
$\lambda$ is small.  A variational state was constructed to yield  the rigorous
upper bound 
\beq
e_0(\rho) \le 4\pi\varrho a \Big [1+
\frac{128}{15\sqrt{\pi}}( \varrho a^3)^{1/2}S_\lambda \Big ]
 + O(\varrho^2|\log \varrho| )  \;
\label{E-1}
\eeq
with $S_\lambda \leq 1+C\lambda$.
In the limit $\lambda \to 0$, one recovers the prediction of Lee-Yang \cite{LYang} and  Lee-Huang-Yang \cite{LHY}. 
The trial state in \cite{ESY}  does not have a
fixed number of particles, and is a state in the Fock space
with expected number of particles  $N$ (Presumably a  trial state with a  fixed number of particles
can be constructed with a similar idea).
The trial state in \cite{ESY} is similar to the trial state used by Girardeau and Arnowitt \cite{GA} 
and  recently by Solovej \cite{Sol}; it is of the form 
\beq\label{trialstate0}
 \exp \left [ |\Lambda|^{-1}  \sum_k  {c_k}  \cre_k\cre_{-k}\ann_0\ann_0  +  \sqrt{N_0}\cre_0  \right  ] |0\rangle
\eeq 
where  $c_k$ and $N_0$  have to be chosen carefully to give the correct asymptotic in energy.  This state captures  the idea that 
particle pairs of opposite momenta  are created from the sea of condensate consisting of  zero momentum particles.  It is believed that this type of trial state gives the ground state energy consistent with the 
Bogoliubov approximation. In the case of Bose gas,  the Bogoliubov approximation yields  
the correct energy up to  the order $\rho^{3/2}$, but the constant is correct only in the semiclassical limit---
consistent with the calculation using the trial state \eqref{trialstate0}.   
It should be noted that the Bogoliubov approximation gives the  correct ``correlation  energy"  in several 
setting including the one and  two component charged Boson gases  \cite{LS1, LS2, Sol} and 
the Bose gas  in  large density-weak potential limit \cite{GS}.

For the Bose gas in low density, the result of \cite{ESY} suggests to correct the error by 
renormalizing the the propagator.  Unfortunately,  it is difficult to implement this idea. 
Our main observation is to relax the concept of condensates 
by allowing particle pairs to have nonzero total momenta. More precisely, we consider a trial state of the form 
\beqa\label{trialstate00}
&& \exp \Big  [ |\Lambda|^{-1}  \sum_k \sum_{v \sim \sqrt \rho }  2 \sqrt { \lambda_{k+v/2}  \lambda_{-k+v/2}} \cre_{k+v/2}\cre_{-k+v/2}\ann_v\ann_0   \nonumber \\ 
&& +  |\Lambda|^{-1}  \sum_k  {c_k}  \cre_k\cre_{-k}\ann_0\ann_0  +  \sqrt{N_0}\cre_0  \Big  ] |0\rangle
\eeqa
for suitably chosen $c$ and $\lambda$. 
Notice that the total momentum of the pair, $v$, is required to be of  order $\rho^{1/2}$ and the constant $2$ comes from the ordering of 
$\ann_v \ann_0$. We shall make further simplification that $\lambda_k = c_k$. Even with this simplification, 
however,  
this state  is still too complicated. We will extract some properties from this representation and define an 
$N$ particle trial state whose energy is given by the Lee-Yang's prediction up to the second order term.  Details will be given in Section \ref{trialstate}.  Our result shows that, in order to obtain the second order energy,  the typical ansatz for the Bogoliubov approximation should be extended to allow pair particles with nonzero  momenta. This also suggests that the Bogoliubov approximation has to be  modified in order to yield the correct energy of the low density  Bose gas to the second order.

\section{Notations  and Main Results}

Let  $\Lambda= [0,L]^3\subset \R^3$ be a cube 
with periodic boundary conditions with the dual space  
$\Lambda^*:= (\frac{2\pi}{L} \Z)^3$. 
The Fourier transform is defined as
$$
      W_p:= \hat W(p)  = \int_{x\in \R^3} e^{-ipx} W(x) \rd x, \qquad
     W(x) = \frac{1}{|\Lambda|} \sum_{p\in \Lambda^*}e^{ipx}  W_p \; .
$$ 
Here we have used the convention to denote the Fourier transform of a function $W$ at the momentum $p$ by $W_p$ 
instead of $\wh W (p)$  to avoid too heavy notations. 
%
Since the summation of $p$ is always restricted to $\Lambda^*$, we will  not explicitly specify it. 
%

%

%
%

We will use the bosonic operators with the commutator relations
$$
      [\ann_p, \cre_q] = \ann_p \cre_q -\cre_q\ann_p
     = \left\{
\begin{array}{ll}
  1 & \mbox{ if } p=q \\ 
0 & \mbox{ otherwise.}
\end{array} 
\right.
$$
The two body interaction is given by  a smooth, symmetric  non-negative function $V(x)$  of fast decay. 
Clearly, in the Fourier space, we have 
$    V_u=V_{-u}= \bar V_u$. 
Furthermore,  we assume that the potential $V$ is small so that the Born series converges.
The  Hamiltonian of the many-body systems with the potential $V$ and the periodic boundary condition is thus  given by
\beq
      H = 
\sum_p p^2 \cre_p\ann_p + \frac{1}{|\Lambda|} \sum_{p,q,u}
     V_u \cre_p\cre_q\ann_{p-u}\ann_{q+u}
\label{ham}
\eeq

Let $1-w$ be the zero energy scattering solution 
$$
    -\Delta(1-w) +  V(1-w)=0
$$
with $0\leq w<1$ and $w(x)\to 0$ as $|x|\to \infty$.
Then the scattering length is given by the formula 
$$
   a: =\frac{1}{4\pi}\int_{\R^3} V(x)(1-w(x))\rd x
$$
Introduce $g_0$, whose meaning will be explained later on,  to denote the  quantity
\[
g_0 = 4 \pi a.
\]

Let $\mathcal H_N$ be the Hilbert space of  $N$ bosons. Denote by $\rho_N=N/\Lambda$ the density of the system. 
The ground state energy of the Hamiltonian \eqref{ham} in  $\mathcal H_N$  
is given by 
\[
E^{P}_0(\rho,\Lambda) = \inf spec H_{\mathcal H_N}
\]
and the ground state energy per particle is $e_0^P(\rho,\Lambda) =  E_0^P(\rho,\Lambda)/N$. We can also consider other boundary 
conditions, e.g., $e_0^D(\rho,\Lambda)$ is the Dirichlet boundary condition ground state energy per particle.

In this paper, we will always take the limit $L \to \infty$ so that the density $\rho_N \to \rho$ 
for some fixed density $\rho$. {F}rom now on,  we will use $\lim_{L \to\infty}$  for the more complicated notation 
$\lim_{L \to \infty, \, N/L^3 \to \rho}$. We now state the main result of this paper.

\begin{thm}\label{mainthm} 
Suppose the potential $V$ is  smooth, symmetric, nonnegative with fast decay
and sufficiently small so that the Born series converges. 
Then the ground state energy per particle  satisfies the upper bound 
\beq
\overline\lim_{\rho\to0} \overline \lim_{L \to\infty}\left(
\frac{e_0^P(\rho,\Lambda)- g_0 \rho }{g_0^{5/2} \rho ^{3/2} }\right)\leq \frac{16}{15\pi^2}
\eeq
\end{thm}
Although we state the theorem in the form of limit $\rho \to 0$, an error bound  is available from the proof. 
We avoid stating such an estimate to simplify the notations and proofs. 
Our result  holds also for Dirichlet boundary condition.

\subsection{Reduction to Small Torus with Periodic Boundary Conditions}

To prove Theorem \ref{mainthm}, we only need to construct a trial state $\Psi(\rho, \Lambda)$ satisfying  the boundary condition and
\beq\label{resultpro1}
\overline\lim_{\rho\to0}\lim_{\Lambda\to\infty}\left(
\frac{\langle H_N\rangle_\Psi N^{-1}-g_0\rho }{ g_0^{5/2} \rho ^{3/2}}\right)\leq \frac{16}{15\pi^2}
\eeq
The first step is to construct a trial state with a Dirichlet boundary condition in a cube  of order slightly bigger than $\rho^{-1}$. 
\begin{lem}\label{thm1}
For density $ \rho$ small enough, there exist $ L \sim  \rho^{\,-25/24}$ and 
 a trial state $\Psi$ of $ N$ $( N= \rho  L^3)$ particles on $ \Lambda = [0,   L]^3$ satisfying  the Dirichlet boundary condition and 
 \beq
 \overline\lim_{ \rho\to0}\left(
\frac{\langle H_{ N}\rangle_\Psi  { N}^{-1}-g_0 \rho }{ g_0^{5/2}  \rho ^{\,3/2}}\right)\leq \frac{16}{15\pi^2}.
 \eeq 
\end{lem}
Once we have a trial state with the Dirichlet boundary condition, we can duplicate it so that a trial state can be constructed for cubes 
with linear dimension $  \gg  \rho^{-25/24}$. This proves Theorem \ref{mainthm}.

The next lemma shows that a Dirichlet boundary condition trial state with correct energy can be obtained from a   
periodic one.

\begin{lem}\label{lemma1}
Recall the ground state energies  per particle  $e^{D}_0(\rho,\Lambda)$ and $e^{P}_0(\rho,\Lambda)$ for 
the Dirichlet and periodic boundary condition. 
Let   $\Lambda = [0, L]^3$ and  $L = \rho^{-25/24}$. Suppose  the energy for the periodic boundary condition 
satisfies that 
\beq
\overline\lim_{\rho\to0}\left(
\frac{e^P_0(\rho,\Lambda)-g_0\rho }{ g_0^{5/2} \rho ^{3/2}}\right)\leq \frac{16}{15\pi^2}
\eeq
Then for $\widetilde\Lambda=[0, \widetilde L ]^3$, $\widetilde L=L(1+2\rho^{25/48})$ and $\widetilde \rho=\rho L^3/\widetilde L^3$, the following estimate for the energy of the Dirichlet boundary condition  holds: 
\beq\label{resultlemma1}
\overline\lim_{\rho\to0}\left(
\frac{e^D_0(\widetilde \rho,\widetilde\Lambda)-g_0\widetilde \rho }{ g_0^{5/2} \widetilde \rho ^{\,3/2}}\right)\leq \frac{16}{15\pi^2}
\eeq
\end{lem}

The construction of  a periodic trial state yielding the correct energy upper bound is the core of this paper. 
We state it as the following theorem. 
 
\begin{thm}\label{lemma2}
There exists a periodic trial state $\Psi$ of $N$ particles on $\Lambda = [ 0, L]^3$,  $L = \rho^{-25/24}$ such that ($N=|\Lambda|\rho$) 
\beq
\overline\lim_{\rho\to0}\left(
\frac{\langle H_N\rangle_\Psi N^{-1}-g_0\rho }{ g_0^{5/2}\rho^{3/2}}\right)\leq \frac{16}{15\pi^2}
\eeq
\end{thm}

This paper is organized as follows: In Section 3, we define rigorously the trial state. 
In Section 4, we outline the Lemmas needed to prove Theorem \ref{lemma2}. 
In Section 5, we estimate the number of particles in the condensate and various momentum regimes. 
These estimates are the building blocks for all other estimates later on. In Section 6, we estimate the kinetic energy. The potential energy is estimated in Section 7-11. Finally in Section 12, we prove the reduction 
to the periodic boundary condition, i.e., Lemma \ref{lemma1}. This proof follows a standard approach 
and only a sketch will be given.

\section{Definition of the Trial State}\label{trialstate}

We now give a formal definition of the trial state. This somehow abstract definition will be explained later on.  
We first identify  four regions in the momentum space  $\Lambda^*$  which are  relevant to the construction of the trial state: $P_0$ for the condensate,  $P_L$ for the low momenta, which are of the order  $\rho^{1/2}$;  $P_H$ for momenta of order one, and $P_I$ the region between 
$P_L$ and $P_H$. 

\begin{mydef} \label{def1} Define four subsets of momentum space: $P_0$, $P_L$, $P_I$ and $P_H$.
\beqa\nonumber
P_0&&\equiv\left\{p=0\right\}\\\nonumber
P_L&&\equiv\left\{p\in \Lambda^*| \eps_L\rho^{1/2}\leq|p|\leq \eta^{-1}_L\rho^{1/2}\right\}\\\nonumber
P_I&&\equiv\left\{p\in \Lambda^*|\eta^{-1}_L\rho^{1/2} <   |p|\leq \eps_H\right\}\\
P_H&&\equiv\left\{p\in \Lambda^*|\eps_H <   |p|\right\}\, , 
\eeqa
where the parameters are chosen so that  
\beq\label{defepseta}
\eps_L, \eta_L, \eps_H\equiv\rho^{\eta}\rmand \eta\equiv1/200
\eeq
Denote by $P = P_0\cup P_L \cup P_I \cup P_H$.
\end{mydef}
We remark that the momenta between $P_0$ and $P_L$ are irrelevant to our construction. 
Next, we need a notation for the collection of states with $N$ particles. 
\begin{mydef}
Let  $\widetilde M$ be  the set of all functions $\al: P \rightarrow \N\cup 0$ such that
\beq
\sum_{k\in P}\al(k)=N
\eeq
For any $\al \in \widetilde M$,  denote by $|\al\rangle\in \mathcal H_N$  the unique state (in this case, an  $N$-particle  wave function) 
defined by the map $\alpha$ 
\[
|\al\rangle =  C \prod_{k\in P} (a^\dagger_k)^{\alpha(k)} | 0 \rangle\, , 
\]
where the positive constant $C$ is chosen so that $|\al\rangle $ is $L_2$ normalized.  
Define $\al_{free}$ as $\al_{free}(k)=N\delta_{0,k}$.
\end{mydef}

Clearly, we have 
\beq\label{number}
a^\dagger_k a_k|\al\rangle=\al(k)|\al\rangle,\,\,\, \forall k\in P
\eeq



%

\begin{mydef}
We define two  relations between functions in $\widetilde M$: 
\begin{enumerate}
	\item Strict pair creation of momentum $k$: Denote by  $
\be: =  \mathcal A^k \al $
 if $\beta$ is generated by creating 
a pair of particles with momenta $k$ and $-k$, i.e., 
\beq	
\be(p)= \left\{ 
\begin{array}{ll}
\al(p)-2,& p=0 \\ 
\al(p)+1, & p=\pm k\\
\al(p), & others
\end{array}\right.
\eeq
In terms of states, we have 
\[
|\be\rangle =  C a^+_k a^+_{-k} a_0^2 | \alpha\rangle 
\]
where $C$ is a positive constant so that the state $|\be\rangle$ is normalized.

\item Soft pair creation with total momentum $u$ and difference $ 2 k$:   
Denote by  $\beta=\mathcal A^{u,\,k}\al$ if $\beta$ is generated by creating two  particles with high momenta $\pm k+u/2\in P_H$ 
so that the total momentum $u$ is in $P_L$, i.e., 
\beq	
\be(p)= \left\{ 
\begin{array}{ll}
\al(p)-1,& p=0\,\,\, {\rm or}\,\,\, u \\ 
\al(p)+1, & p=\pm k+u/2\\
\al(p), & others
\end{array}\right.
\eeq
Notice that $\mathcal A^{u,\,k}\al$ is defined only if $\pm k+u/2\in P_H$. 
In terms of states, we have 
\[
|\be\rangle = C a^+_{k+ u/2} a^+_{-k + u/2} a_0 a_u | \alpha\rangle 
\]
where $C$ is the normalization constant. 
Since $\beta(p)$ 
has to be  nonnegative,  the state  $\mathcal A^{k}\al$ or  $\mathcal A^{u,\,k}\al$ is not defined for all $\alpha$ or $k, u$. 
\end{enumerate}
Define $D_\alpha$ to be  the set all possible derivations of $\alpha$ from the previous two operations: 
\beq
D_\alpha=\left\{\mathcal A^{u,\,k}\al \in \widetilde M \right\}\cup\left\{\mathcal A^{k}\al \in \widetilde M\right\}
\eeq
\end{mydef}

Our trial state will be of the form  $\sum_{\al\in \widetilde  M}f(\al)|\al\rangle$ where $f$ is  supported in a 
subset of $\widetilde M$ which we now define.

\begin{mydef}\label{M}
Fix a large real number $k_c$. 
We define $M$ as the smallest subset of $\widetilde M$ such that
\begin{enumerate}
	\item $\al_{free} \in M$.
	\item  $M$ is closed under strict pair creation provided the momentum  $u\in P_I\cup P_H$, i.e.,  if  $\al \in M$ and $\mathcal A^{u}\al\in \widetilde M$ then $\mathcal A^{u}\al\in M$.
	\item $M$ is closed under strict pair creation provided the momentum  $u\in P_L$  and  $\max\{\al(u), \al(-u)\}<\mcut$, i.e., 
	if  $\al \in M$ and $\mathcal A^{u}\al\in \widetilde M$, then $\mathcal A^{u}\al\in M$. Here we choose $m_c$ as 
	\beq\label{defmc}
	m_c\equiv\rho^{-\eta}=\rho^{-1/200}
	\eeq
	\item   M is closed under soft pair creation from states with perfect pairing of momenta $u$ and $-u$. More precisely, 
	for $u \in P_L$ with $\al(u)=\al(-u)$, if $\al \in M$, $\mathcal A^{u,\,k}\al\in \widetilde M$ and 
	$$ \e_H \le |\pm k+u/2|\leq k_c,
	$$ 
	then $\mathcal A^{u,\,k}\al\in M$. 
\end{enumerate}
\end{mydef}
The set $M$ is unique since the intersection of two such sets $M_1$ and $M_2$ satisfies all four conditions. 

For any $u\in P_L$, we define the set of states with symmetric (asymmetric resp.)  pair particles of momenta $u, -u$ 
by  $M_u^s$ ($M_u^a$ resp.): 
\beqa\label{defMsau}
&&M_u^s\equiv\{\al\in M|\al(u)=\al(-u)\}\\\nonumber
&&M_u^a\equiv\{\al\in M|\al(u)\neq\al(-u)\}.
\eeqa
Denote by  $\al^*(u)$  the maximum of $\al(u)$ and $ \al(-u)$:
\beq\label{def*}
\al^*(u)=\max\{\al(u), \al(-u)\}
\eeq 
Since soft pair creation was allowed only from momenta in $P_L$ and the final momenta are in $P_H$, we have 
\[
\al^*(u)- \al(u)\in \{0, 1 \}, \quad  \al(-u)= \al(u), \text{ for all }  u \in P_I
\]

Before defining the weight $f(\alpha)$, we introduce several quantities related 
to the scattering equation. In the momentum space, the  scattering equation is 
given by ($p\in \R^3$)
%
\beq\label{av}
   -p^2 w_p+ V_p -\int_r V_{p-r} w_r=0, \qquad \forall p \neq 0
\eeq 
 Let $g$ be the function 
\beq
     g(x): = V(x) (1- w(x))
\label{gdef}
\eeq
Then the scattering equation in momentum space takes the form
\beq\label{gp}
     g_p= p^2w_p  \qquad \forall p \neq 0
\eeq
One can check $ 4\pi a =g_0$  
this explains the notation $g_{0}$ used in Theorem \ref{mainthm} and Theorem \ref{lemma2}. 

\begin{mydef}
Define  for all $ \eps\neq 0$
\beq\label{defrho0}
\rhop\equiv\rho_0+\eps\rho^{3/2}\, , \, \rho_0:= \rho-\frac{1}{3\pi^2}(g_0)^{3/2}\rho^{3/2}, 
\eeq
where  $\rho_0$ will be the approximate density of the condensate.
Define the ``chemical potential"   $\lambda$ by 
\beq\label{lambda}
\lambda_k= \left\{ 
\begin{array}{ll}
\frac{1-\sqrt{1+4\rho g_0|k|^{-2}}}{1+\sqrt{1+4\rho g_0|k|^{-2}}}\,\rho^{-1},& k\in P_L \\ 
-w_k, & k\in P_I\cup P_H
\end{array}\right.
\eeq
\end{mydef}
One can check that, to the leading order, $\lambda$ is given by 
\beq\label{lambda0}
\rho\lambda_k\equiv\frac{1-\sqrt{1+4\rho g_k|k|^{-2}}}{1+\sqrt{1+4\rho g_k|k|^{-2}}}
\eeq
Notice that $\lambda_k$ is real number and can be negative.

\begin{mydef}{\bf The Trial State}
\par Let  $\Psi$ be defined by 
\beq
\Psi\equiv\sum_{\al\in M}f(\al)|\al\rangle
\eeq
where the coefficient $f$ is given by 
 \beq\label{deffal}
f(\al)=C_N 
\sqrt { \frac  {|\Lambda|^{\alpha(0)}} {\al(0)!} }\prod_{k\neq 0}(\sqrt{\lambda_{k}})^{\al(k)}\prod_{u\in P_L, \alpha^\ast (u) - \alpha (u) = 1}\sqrt{\frac{4\al^*(u)\lambda_u}{|\Lambda|}}
\eeq
Here we follow the convention $\sqrt{x}=\sqrt{|x|}i$ for $x<0$. For convenience, we define $f(\al)=0$ for $\al\notin M$.
The constant  $C_N$ is  chosen so that $\Psi$ is $L_2$ normalized, i.e., 
\[\left\langle \Psi|\Psi\right\rangle=1.\]  
\end{mydef}

\begin{thm}\label{lemma2.1} Suppose  $\Lambda = [ 0, L]^3$ and $L = \rho^{-25/24}$. 
Then the  trial state $\Psi$ in \eqref{deffal} satisfies the  estimate 
\beq
\overline\lim_{k_c \to \infty} \overline\lim_{\rho\to0}\left(
\frac{\langle H_N\rangle_\Psi N^{-1}-g_0\rho }{ g_0^{5/2}\rho^{3/2}}\right)\leq \frac{16}{15\pi^2},
\eeq
where $k_c$  is given in Definition \ref{M}. We recall that $m_c^{-1}, \e_L, \eta_L, \e_H$ are chosen as a small power of $\rho$ in \eqref{defepseta} and \eqref{defmc}. 
\end{thm}

\subsection{Heuristic Derivation of the Trial State}

We now give a heuristic idea for the construction of  the trial state. Fix an ordering of momenta in $\Lambda^*$ so that the first one 
is the zero momentum.  We will  use the occupation number representation so that 
\beq
|n_1,\,\,n_2,\,\,\cdots\rangle
\eeq
represents the normalized state with  $n_i$ particles of momentum $k_i$. For example, 
\[
|N,\; 0, 0 ,\,\cdots\rangle = \frac 1 { \sqrt {N!}}( \cre_0)^N |0\rangle 
\]
Recall that we would like to generate a state of the form in \eqref{trialstate00}. A slightly modified one is 
\beqa\label{trialstate001}
&& \exp \Big  [ |\Lambda|^{-1}  \sum_k \sum_{v \sim \sqrt \rho }  2 \sqrt {\lambda_{k+v/2}  \lambda_{-k+v/2}} \cre_{k+v/2}\cre_{-k+v/2}\ann_v\ann_0   \nonumber \\ 
&& + |\Lambda|^{-1} \sum_k  {\lambda_k}  \cre_k\cre_{-k}\ann_0\ann_0   \Big  ] |N,\; 0, 0 ,\,\cdots\rangle
\eeqa

We now expand the exponential and require that $\cre_{k+v/2}\cre_{-k+v/2}\ann_v\ann_0$ to  appear at most once. 
The rationale of this assumption is that the soft pair creation is a rare event and thus we can neglect higher order terms.  
Our trial state is thus a sum  of the following  state parametrized by  $k_1,\cdots,k_s$, $n_1,\cdots,n_s$, $k'_1,\cdots,k'_t$
and $v_1,\cdots,v_t$: 
\beq\label{defmn}
 \const \prod_{j=1}^t    \sqrt { 4 \lambda_{k'_j+v_j/2}  \lambda_{-k'_j+v_j/2}}  \prod_{i=1}^s\left(\lambda_{k_i} \right)^{n_i} |\al\rangle
\eeq
where 
\beqa
|\al\rangle = && \const |\Lambda|^{-t-\sum_{i=1}^s n_i}  \prod_{j=1}^t    \cre_{\frac{v_j}{2}+k'_j}\cre_{\frac{v_j}{2}-k'_j}\ann_{v_j}  \ann_0  \nonumber   \\
&& \times \prod_{i=1}^s   \frac 1 {n_i !} \left ( \cre_{k_i}\cre_{-k_i}\ann_{0}\ann_0\right)^{n_i}|N,0,\cdots\rangle
\eeqa
Here we have chosen the constant so that the norm of $|\al\rangle$ is one. 
We  also require that  $v_i+v_j\neq0$ for $1\leq i,j\leq t$ since $v_i+v_j= 0$  is a higher order event. 

We  further make  the simplifying  assumption that $v_i \in P_L$.  
Observe now that the state $|\al\rangle$ can be obtained from strict and soft pair creations. This explains the core  idea behind the definition of  $M$ in Definition \ref{M}. Other restrictions in the definition were mostly due to various cutoffs needed in the estimates. 
Finally, up to factors depending only on $\Lambda$ and $N$, the coefficient  in \eqref{defmn} gives $f(\alpha)$ in \eqref{deffal}. Notice all factors depending on $s, t, n_i$ were already included in $|\alpha\rangle $.

The choice of $\lambda$ is much more complicated. To the first approximation, $\lambda$ can be obtain
from the work of \cite{ESY}. We thus use this choice to identify the error terms.  Once this is done, we
optimize the main terms and this leads to the current definition of $\lambda$. Notice that, since our trial state
is different, there are more main terms than in \cite{ESY}.


\section{Proof of Theorem \ref{lemma2}}
\begin{proof}
Our goal is to prove 
\beq\label{Maindesired}
\overline
\lim_{k_c\to\infty}
  \left(\overline\lim_{\rho\to0}
     \left(
         \frac{|\Lambda|^{-1}\langle H\rangle_\Psi-g_0\rho^2}
         {\rho^{5/2}}
      \right)
  \right)
\leq \frac{16}{15\pi^2}g_0^{5/2}
\eeq
Here $g_0=4\pi a$, $\langle H\rangle_\Psi=\langle\Psi|H|\Psi\rangle$. We decompose the Hamiltonian as follows:
\beq
H=\sum_{i=1}^N-\Delta_i+H_{S1}+H_{S2}+H_{S3}+H_{A1}+H_{A2},
\eeq
where 
\begin{enumerate}
	\item $H_{S1}$ is the part of interaction that annihilates two particles and creates the same two particles, i.e.,
	\beq
	H_{S1}= |\Lambda|^{-1} \sum_{u}V_{0}a^\dagger_u a^\dagger_u a_u a_u+ |\Lambda|^{-1} \sum_{u\neq v}(V_{u-v}+V_{0})a^\dagger_u a^\dagger_v a_u a_v
	\eeq
	\item $H_{S2}$ is the interaction between the condensate and strict pairs, i.e., 
	\beq
 H_{S2}=	|\Lambda|^{-1}\sum_{u\neq 0}V_{u}a^\dagger_u a^\dagger_{-u} a_0 a_0+C.C.
	\eeq
	\item $H_{S3}$ is the part of interaction that strict pairs are involved, i.e., 
	\beq
	H_{S3}=|\Lambda|^{-1} \sum_{u, v\neq 0, u\neq v} V_{u-v}a^\dagger_u a^\dagger_{-u} a_v a_{-v}
	\eeq
	\item $H_{A1}$ is the part of the  interaction that one and only one condensate particle is involved i.e., 
	\beq
	H_{A1}=|\Lambda|^{-1} \sum_{v_1,v_2,v_3 \neq 0}2V_{v_2}a^\dagger_0 a^\dagger_{v_1} a_{v_2} a_{v_3}+C.C.
	\eeq
	\item $H_{A2}$ is the part of the  interaction which is not counted in $H_{S1}$ and there is no condensate nor strict pair involved i.e., 
	\beq
	H_{A2}= |\Lambda|^{-1} \sum_{v_i \neq 0, v_1+v_2\neq 0, \{v_1,v_2\}\neq \{v_3,v_4\}}V_{v_1-v_3}a^\dagger_{v_1} a^\dagger_{v_2} a_{v_3} a_{v_4}
	\eeq
\end{enumerate} 
The estimates for the energies of these components are stated as the following lemmas, which will be proved 
in later sections. 

\begin{lem}\label{lemkinetic}The total kinetic energy is bounded  above by
\beq\label{proofthem1-1}
\overline\lim_{k_c,\rho}
\left(\frac{1}{|\Lambda|}\left\langle\sum_{i=1}^N-\Delta_i\right\rangle_{\Psi}-\rho_0^2 \|\nabla w\|_2^2\right)\rho^{-\frac{5}{2}}
\leq \frac{4\|\nabla w\|_2^2g_0^{3/2}}{3\pi^2}  -\frac{8g_0^{5/2}}{5\pi^2}
\eeq
\end{lem}
\begin{lem}\label{lemHS1}
The expectation value of $H_{S1}$ is bounded  above by,
\beq\label{proofthem1-2}
\overline\lim_{k_c,\rho}
\left(\frac{1}{|\Lambda|}\left\langle H_{S1}\right\rangle_{\Psi}-\rho_0^2 V_0\right)\rho^{-5/2}
\leq \frac{4V_0g_0^{3/2}}{3\pi^2}
\eeq
\end{lem}
\begin{lem}\label{lemHS2}
The expectation value of $H_{S2}$ is bounded above by,
\beq\label{proofthem1-3}
\overline\lim_{k_c,\rho}
\left(\frac{1}{|\Lambda|}\left\langle H_{S2}\right\rangle_{\Psi}+2\rho_0^2 \|Vw\|_1\right)\rho^{-5/2}
\leq \frac{2V_0g_0^{3/2}}{\pi^2}
\eeq
\end{lem}
\begin{lem}\label{lemHS3}
The expectation value of $H_{S3}$ is bounded above by,
\beq\label{proofthem1-4}
\overline\lim_{k_c,\rho}
\left(\frac{1}{|\Lambda|}\left\langle H_{S3}\right\rangle_{\Psi}-\rho_0^2 \|Vw^2\|_1\right)\rho^{-5/2}
\leq \frac{-2\|Vw\|_1g_0^{3/2}}{\pi^2}
\eeq
\end{lem}
\begin{lem}\label{lemHAS1}
The expectation value of $H_{A1}$ is bounded above by,
\beq\label{proofthem1-5}
\overline\lim_{k_c,\rho}
\left(\frac{1}{|\Lambda|}\left\langle H_{A1}\right\rangle_{\Psi}\right)\rho^{-5/2}
\leq \frac{-8\|Vw\|_1g_0^{3/2}}{3\pi^2}
\eeq
\end{lem}
\begin{lem}\label{lemHAS2}
The expectation value of $H_{A2}$ is bounded above by,
\beq\label{proofthem1-6}
\overline\lim_{k_c,\rho}
\left(\frac{1}{|\Lambda|}\left\langle H_{A2}\right\rangle_{\Psi}\right)\rho^{-5/2}
\leq \frac{4\|Vw^2\|_1g_0^{3/2}}{3\pi^2}
\eeq
\end{lem}

By definitions of $g_0$ and $w$ \eqref{av}, \eqref{gdef}, we have
\beq
 \|\nabla w\|_2^2-\|Vw\|_1+\|Vw^2\|_1=0, \;  
V_0-\|Vw\|_1=g_0
\eeq
Summing  \eqref{proofthem1-1}-\eqref{proofthem1-6}, we have 
\beq
\overline\lim_{k_c,\rho}
\left(\frac{1}{|\Lambda|}\left\langle H_{N}\right\rangle_{\Psi}-\rho^2_0g_0\right)\rho^{-5/2}
\leq \frac{26g_0^{5/2}}{15\pi^2}
\eeq
By definition of $\rho_0$ \eqref{defrho0}, we have proved (\ref{Maindesired}).
\end{proof}

\section{Estimates on the Numbers of Particles}

The first step to prove the Lemma \ref{lemkinetic} to Lemma \ref{lemHAS2} is to estimate the number of 
particles in the condensate, $P_{L}, P_{I}$, and $P_{H}$. This is the main task of this section and 
we start with  the following  notations.

\begin{mydef}\label{1}
Suppose $u_i, k_j \in P$ for $i = 1, \ldots t, j = 1, \ldots, s$.
\begin{enumerate}
	\item  The expectation of the  product of  particle numbers with momenta  $u_1$, $\cdots$ $u_s$: 
	\beqa\nonumber
&&Q_\Psi  \left(u_1,u_2,\cdots,u_s\right)
=\left\langle\prod_{i=1}^sa^\dagger_{u_i} a_{u_i}\right\rangle_{\Psi}
=\sum_{\al\in M} \prod_{i=1}^s{\al(u_i)|f(\al)|^2}
\eeqa

\item The probability to have $m_i$ particles with momentum 
	$u_i, i= 1 \ldots, s$: 
\beqa\label{2.1}
&& Q_\Psi\left(\umt\right)
\equiv\sum_{\al\in A} |f(\al)|^2\\\nonumber
 &&{\rm Here}\,\,\,A=\left\{\al\in M| \al(u_1)=m_1,\cdots,\al(u_t)=m_t\right\}
\eeqa
\item The expectation of the  product of  particle numbers with momenta  $k_1$, $\ldots$, $k_s$,
conditioned that there are  $m_i$ particles with momentum $u_i$:
\beqa\nonumber
&&Q_\Psi\left(k_1,\cdots,k_s \, | \, \umt\right)\\\nonumber
\equiv&&\left(\sum_{\al\in A} \prod_{i=1}^s{\al(k_i)|f(\al)|^2}\right)
\left(\sum_{\al\in A} |f(\al)|^2\right)^{-1}, \\\nonumber
\eeqa
where $A$ is the same as in item 2. 
\end{enumerate}
\end{mydef}

The following theorem provides the main estimates on the number of particles.

\begin{thm}\label{thmA} 
In the limit $\lim_{k_c \to\infty}\lim_{\rho\to0}$, $  Q_\Psi(u)$  can be estimated as follows
\beqa\label{boundspsiuPMPHtotal}
&&\lim_{k_c \to\infty}\lim_{\rho\to0} \left(\rho^{-3/2}|\Lambda|^{-1}\sum_{u\in P_I\cup P_H} Q_\Psi(u)\right)=0\\
\label{boundspsiuPLtotal}
&&\lim_{k_c \to\infty}\lim_{\rho\to0}\left(\rho^{-3/2}|\Lambda|^{-1}\sum_{u\in P_L} Q_\Psi(u)\right)=\frac1{3\pi^2}g^{3/2}_0
\eeqa
\end{thm}

We first collect a few obvious identities of $f$ into the following lemma. 

\begin{lem} \label{f}
\begin{enumerate}
	\item  If $k\in P_I\cup P_H$ and $\al, \mathcal A^k \al  \in M$, then 
	\beq \label{properf1}f( \mathcal A^k \al )=\sqrt{\frac{\al(0)}{|\Lambda|}}\sqrt{\frac{\al(0)-1}{|\Lambda|}}\lambda_k f(\al)
	\eeq
	\item If $k\in P_L$, $\al\in M_k^s$ and $\al, \mathcal A^k \al  \in M$, then 
	\beq \label{properf2}f( \mathcal A^k \al )=\sqrt{\frac{\al(0)}{|\Lambda|}}\sqrt{\frac{\al(0)-1}{|\Lambda|}}\lambda_k f(\al)
	\eeq
  \item If $k\in P_L$, $\al\in M^a_k$ and $\al, \mathcal A^k \al  \in M$, then
 	\beq \label{properf3}f( \mathcal A^k\al )=\sqrt{\frac{\al(0)}{|\Lambda|}}\sqrt{\frac{\al(0)-1}{|\Lambda|}}\sqrt{\frac{\al^*(k)+1}{\al^*(k)}}\lambda_k f(\al)
	\eeq
	\item If $\al\in M_u^s$ and $ \mathcal A^{u, k}\al  \in M$, then 
	\beq\label{properf4}
f( \mathcal A^{u, k}\al )=2\sqrt{\frac{\al(0)}{|\Lambda|}}\sqrt{\frac{\al(u)}{|\Lambda|}}\sqrt{\lambda_{k+\frac{u}{2}}}\sqrt{\lambda_{-k+\frac{u}{2}}}f(\al)
	\eeq
	\item If $\al\in M_u^a$ and $ \mathcal A^{u, k}\al  \in M$, then 
	\beq\label{properf5}
f( \mathcal A^{u, k}\al )=\frac{1}{2\lambda_u}\sqrt{\frac{\al(0)}{|\Lambda|}}\sqrt{\frac{|\Lambda|}{\al(u)}}\sqrt{\lambda_{k+\frac{u}{2}}}\sqrt{\lambda_{-k+\frac{u}{2}}}f(\al)
	\eeq
\end{enumerate}
\end{lem}

In defining the space $M$,   the operation $ \mathcal A^{u, k}\al $  is not allowed when  $\al\in M_u^a$.  However, it is possible through rare coincidences  that  $ \mathcal A^{u, k}\al \in M$ even if $\al\in M_u^a$.  Clearly,   $\al\in M_u^a$ and $ \mathcal A^{u, k}\al  \in M $ imply that $\alpha(u)  = \alpha(-u) + 1$. 
The following lemma summarizes some properties we need for $\lambda$. 
\begin{lem}\label{lem-l}
\begin{enumerate}
	\item For any $k\in P_L\cup P_I\cup P_H$, $\lambda_k$  only depends on $|k|$ and 
	\beq\label{ublambda}
	|\lambda_k|\leq g_k|k|^{-2}\leq g_0|k|^{-2}, \,\,\, |\rho\lambda_k|\leq 1-\const\eps_L 
	 \eeq
	 \item For any $k\in P_L$, $\lambda_k$ is negative and 
	 \beq\label{lublambdaPL0}
	 -\frac{g_0}{2}\eta_L^2\rho^{-1}\geq \lambda_u\geq-\rho^{-1}
	 \eeq
	 \item For any $k\in P_H$, $|\lambda_k|$ is bounded as  
	 \beq\label{ublambdaPH0}
	|\lambda_u|\leq g_0\eps_H^{-2}
	 \eeq
\end{enumerate}
\end{lem}

%

%


To prove Theorem \ref{thmA}, we start with the following estimate on the condensate.

\begin{lem}\label{boundspsi0}
For any $\eps> 0$, when $\rho$ is small enough, 
the expected 
number of zero-momentum particles can be estimated by 
\beqa\label{boundspsi01}
&&|\Lambda|\rhom\leq Q_\Psi\left(0 \right)\leq |\Lambda|\rhop 
\eeqa
\end{lem}

\subsection{A Lower Bound on the Number of Condensates}

Since the total number of particles in fixed to be $N$, upper bound on $Q_\Psi\left(u \right)$ for ($u\neq 0$)
yields   a lower bound for $Q_\Psi\left(0 \right)$.  
%
The following lemma provides the upper bounds for expected number of particles in various momentum space 
regions.

\begin{lem} \label{lem1} For small enough $\rho$,  
the following upper bounds on $ Q_\Psi(u)$ hold:
\begin{enumerate}
	 \item For $u\in P_I$, 
	 \beq\label{ubpsiuroughPM}
 Q_\Psi(u)\leq \frac{\lambda_u^2\rho^2}{1-\lambda_u^2\rho^2}=\sum_{i=1}^\infty(\lambda_u\rho)^{2i}
\eeq
\item For $u\in P_L$, 
\beq\label{ubpsiuroughPL}
 Q_\Psi(u)\leq  \frac{\lambda_u^2\rho^2}{1-\lambda_u^2\rho^2}\left(1+\const\frac{\rho m_c}{\eps_H}\right)
\eeq
\item For $u\in P_H$,  
	 \beq\label{ubpsiuroughPH}
 Q_\Psi(u)\leq\const \rho^2|u|^{-2}|\lambda_u|
\eeq
\end{enumerate}
\end{lem}

\begin{proof}  The basic idea to prove Lemma \ref{lem1} is the following lemma which  compares, in particular,
$Q_\Psi(\{u,m\})$ and $Q_\Psi(\{u,m-1\})$.

\begin{prop}\label{prop5}
When $\rho$ is small enough,  for any $u \in P_I$, we have
\beq\label{7}
 Q_\Psi(\{u,m\})\leq (\lambda_u\rho)^{2i} Q_\Psi(\{u,m-i\}){\rm \,\,\,for\,\,\,} m\geq i \ge 1
\eeq
\end{prop}

\begin{proof}
We start with the following simple observation, whose proof is obvious and we omit it.  

\begin{prop} \label{prop1}  For any $u \in P_I$ fixed and  all  
$\al\in M$ with $\al(u)=m\ge 1$,  there exists  a $\be\in M$ such that  $\mathcal A^{u}\beta=\al$ and $\be(u)=m-1$. 
\end{prop}

{F}rom the property of $f$ in (\ref{properf1}) and $\be(0)\leq N$,  we obtain 
\[
|f(\mathcal A^{u}\beta)|= |\lambda_u|\frac{\be(0)}{|\Lambda|}|f(\be)|\leq |\lambda_u|\rho\,|f(\be)|
\]
Therefore,  we have for $m\ge 1$
\beqa\label{ubpsiuroughPM1}
Q_\Psi\left(\{u,m\}\right) &&\leq \sum_{\be(u)=m-1}|f(\mathcal A^{u}\beta)|^2\leq\lambda_u^2\rho^2\sum_{\be(u)=m-1}|f(\be)|^2
\\\nonumber
&&=\lambda_u^2\rho^2 Q_\Psi\left(\{u,m-1\}\right)
\eeqa
This proves \eqref{7} for $i = 1$. The general cases follow from iterations.

\end{proof}

Together with  $\sum_{m=0}^N Q_\Psi\left(\{u,m\}\right)=1$, we have 
\beqa\label{4}
 Q_\Psi(u)=&&\sum_{m=1}^N m  Q_\Psi(\{u,m\})=\sum_{i=1}^N\left(\sum_{m=i}^N
Q_\Psi\left(\{u,m\}\right)\right)\\\nonumber
\leq &&\sum_{i=1}^N(\lambda_u\rho)^{2i}\left(\sum_{m=0}^N Q_\Psi\left(\{u,m\}\right)\right)=\frac{\lambda_u^2\rho^2}{1-\lambda_u^2\rho^2}
\eeqa
This proves  \eqref{ubpsiuroughPM}.

\par We now prove (\ref{ubpsiuroughPL}).  Recall that $\rho$ is small,  $1\leq m\leq m_c$ and $u\in P_L$.
{F}rom the definition of $M$ \eqref{defMsau},  all elements in the asymmetric part,  $M^a_u$, are generated from the symmetric part
$M^s_u$ via soft pair creations. Thus 
\beq\label{ubpsiuroughPL1.5}
\sum_{\al:\, \al\in M_u^a}^{\al^*(u)=m}|f(\al)|^2\leq\sum_{\be:\, \be\in M_u^s}^{\be(u)=m}\left(\sum_{k: \pm k+u/2\in P_H}|f(\mathcal A^{u, k} \beta)|^2\right)
\eeq
{F}rom \eqref{properf4}, we have,  for $\beta(u) \le m$, 
\beqa\nonumber
|f(\mathcal A^{u, k} \beta)|^2&=&  {4} \left|\lambda_{k+u/2}\lambda_{-k+u/2}\right|\frac{\be(0)}{|\Lambda|}\frac{\be(u)}{|\Lambda|}|f(\be)|^2\\\label{ubpsiuroughPL2}
&\leq&   {4} \left|\lambda_{k+u/2}\lambda_{-k+u/2}\right|\frac{\rho m}{|\Lambda|}|f(\be)|^2
\eeqa
Using the upper bound of $\lambda_k$ in (\ref{ublambda}) and  $|u|\ll |k|$, we have
\beq
\sum_{k:\pm k+u/2\in P_H}\left|\lambda_{k+u/2}\lambda_{-k+u/2}\right|\leq \sum_{p\in P_H}\const|p\,|^{-4}\leq \const\eps_H^{-1} |\Lambda|
\eeq
Inserting these results into \eqref{ubpsiuroughPL1.5}, we obtain 
\beq\label{ubpsiuroughPL1}
\sum_{\al:\, \al\in M_u^a, \al^*(u)=m}|f(\al)|^2\leq\const\frac{\rho m}{\eps_H} \sum_{\be:\, \be\in M_u^s,  \be(u)=m}|f(\be)|^2
\eeq
Summing the last bound over $1 \le m \le m_c$, we have,  for each $u$ fixed, 
\beq\label{boundMa}
\sum_{\al:\, \al\in M_u^a}|f(\al)|^2\leq\const\frac{\rho m_c}{\eps_H}
\eeq
Using this method, we can also prove, for $u\neq \pm v$, 
\beq\label{boundMaMa}
\sum_{\al:\, \al\in M_u^a, \al\in M_v^a}|f(\al)|^2\leq\const(\frac{\rho m_c}{\eps_H})^2
\eeq
{F}rom \eqref{ubpsiuroughPL1}, we have, for  $\rho$ is small enough
\beq
 Q_\Psi(u)\leq \sum_{m=1}^{m_c} \left(m\sum_{\al:\,\al\in M^s_u}^{\al(u)=m}|f(\al)|^2\right)\left(1+\const\frac{\rho m_c}{\eps_H}\right)
\eeq

Following the proof of (\ref{ubpsiuroughPM1}), we have the bound 
\beq\label{pubrho3}
\left[\sum_{\al:\al(u)=m}^{\al\in M^s_u}|f(\al)|^2\right]\leq \lambda_u^2\rho^2\left[\sum_{\be:\be(u)=m-1}^{\be\in M^s_u}|f(\be)|^2\right]
\eeq
Therefore, we can  prove (\ref{ubpsiuroughPL}) using the  argument of \eqref{4}.

\bigskip

\par  We now prove  (\ref{ubpsiuroughPH}) by starting with the following proposition. Once again, the proof is 
straightforward and we omit it. 

\begin{prop}\label{prop2}
For any $u \in P_H$ fixed and  all  
$\al\in M$ with $\al(u)=m\ge 1$,  either there exists   $\be\in M$ such that  $\mathcal A^{u}\beta=\al$ and $\be(u)=m-1$
or there exists $v \in P_L$ and $\be\in M_v^s$  such that  $ \alpha = \mathcal A^{v,\,u-v/2}  \beta$. 
\end{prop}

{F}rom this proposition,  we have 
\beq\label{4.1}
 Q_\Psi(\{u,m\})\leq \sum_{\be:\,\be(u)=m-1}\left[\left|f(\mathcal A^{u}\beta)\right|^2+\sum_{v \in P_{L}, \mathcal A^{v,\,u-v/2}  \beta\in M}^{\be\in M_v^s}\left|f(\mathcal A^{v,\,u-v/2}  \beta)\right|^2\right].
\eeq
By  the properties of $f$ in (\ref{properf1}, \ref{properf4}), we obtain 
\[
|f( \mathcal A^u\be )|^2+\sum_{v\in P_{L} }|f(\mathcal A^{v,\,u-v/2}  \beta)|^2
\leq \left(\rho^2\lambda_u^2+\sum_{v\in P_{L}}4\rho\frac{\be(v)}{|\Lambda|}\left|\lambda_{u}\lambda_{-u+v}\right|\right)|f(\be)|^2 .
\]
Since  $v \in P_{L}$ and $ u \in P_{H}$, from (\ref{ublambda}) we have $\left|\lambda_u\right|$, 
$\left|\lambda_{-u+v}\right| \le \const |u|^{-2} $.  By definition of $M$, $\beta(v) \le m_{c}$. Thus 
\[
\sum_{v\in P_{L}}4\rho \frac{\be(v)}{|\Lambda|}
\le \sum_{v\in P_{L}}4\rho\frac{m_{c}}{|\Lambda|} \le \const\eta_{L}^{-3}m_c\rho^{5/2}.
\]
Hence  we have 
\[
|f( \mathcal A^u\be )|^2+\sum_{v\in P_{L}}|f(\mathcal A^{v,\,u-v/2}  \beta)|^2 \leq \const |u|^{-2}\left|\lambda_u\right|\rho^2|f(\be)|^2.
\]
Together with the bound in \eqref{4.1}, we obtain 
\beq\label{temp6.29}
 Q_\Psi(\{u,m\})\leq \const |\lambda_u| |u|^{-2}\rho^2 Q_\Psi(\{u,m-1\}){\rm\,\,\,for\,\,\,} m\geq 1 .
\eeq
Summing the last inequality over $m$, we have proved  (\ref{ubpsiuroughPH}).

 \end{proof}


The summations in the inequalities in Lemma \ref{lem1} can be performed; 
we summarize the conclusions in the following lemma. 

\begin{prop}\label{particlebound} 
Recall that $\e_L, \eta_L, \e_H$ are chosen in Definition \ref{def1}
as $\rho^\eta$. Then for any $k_c$ and  small enough $\rho$ we have  
\beqa\label{temp4.17}
&&|\Lambda|^{-1}\sum_{u\in P_I} Q_\Psi(u)  \leq \const\rho^{3/2+\eta}
\\
&&|\Lambda|^{-1}\sum_{u\in P_H} Q_\Psi(u)\leq \rho^{7/4} \label{4.18}
\\
&&|\Lambda|^{-1}\sum_{u\in P_L} Q_\Psi(u) \leq \left(\frac{g_0^{3/2}}{3\pi^2}+\const\rho^\eta\right)\rho^{3/2}
\label{4.19}
\eeqa
\end{prop}

Assuming this proposition, we have, for any $\eps>0$, when $\rho$ is small enough, 
\beq\label{resultlbpsi0}
Q_\Psi(0)=N-\sum_{u\neq 0} Q_\Psi(u)
\geq \rhom |\Lambda|
\eeq
 This  proves the lower bound in Lemma \ref{boundspsi0}.  We now prove Proposition \ref{particlebound}.

\begin{proof} The upper bound \eqref{4.18} follows from \eqref{ubpsiuroughPH}, 
$|\lambda_u|\leq g_0|u|^{-2}$ \eqref{ublambda} and the assumption  $u \ge \e_H$ for $u \in P_H$.  

To prove the other bounds, we  first  sum over $u \in P_L$ in \eqref{ubpsiuroughPL} to have 
\beq
|\Lambda|^{-1}\sum_{u\in P_L} Q_\Psi(u)\leq    |\Lambda|^{-1} 
\sum_{u\in P_L}\frac{(\rho\lambda_u)^2}{1-(\rho\lambda_u)^2}(1+\rho^{3/4}), 
\eeq
where we have bounded the factor $\rho m_c /\e_H$ in the error term by $\rho^{3/4}$. 

Let $h(k)=\sqrt{1+4g_0|k|^{-2}}$ and we can rewrite $\lambda$ as 
\beq \label{defhk}
\rho\lambda_{\sqrt {\rho} k} = 
 \frac {1 - h(k)} { 1 + h(k) }   .
 \eeq
Recall for any continuous function $F$ on $\R^3$, we have
$$
    \frac{1}{L^d} \sum_{p\in\Lambda^*} F(p) =
   \frac{1}{|\Lambda|} \sum_{p\in\Lambda^*} F(p)
   \to  \int_{\R^3} 
\frac{{\rm d}^3 p}{(2\pi)^3} F(p)
$$
Thus we have 
\beqa\label{3.1}
&&\lim_{\rho\to 0}|\Lambda|^{-1}\rho^{-3/2}\left(\sum_{u\in P_L}\frac{(\rho\lambda_u)^2}{1-(\rho\lambda_u)^2}\right) \nonumber \\
=  &&\lim_{\rho\to 0}\frac{1}{(2\pi)^3}
\int_{\eps_L\leq |k|\leq \eta_L^{-1}}
\frac{(h(k)-1)^2}{4h(k)}dk^3 +  O(|\Lambda|^{-1/3})  .
\eeqa
The last error comes from replacing the summation by integral.

Due to the choices of  $\eps_L, \eta_L$, 
we can continue the computation as
\beqa\label{sum}
  &&\lim_{\rho\to 0}\left(\frac{1}{(2\pi)^3}
\int_{\eps_L\leq |k|\leq \eta_L^{-1}}
\frac{(h(k)-1)^2}{4h(k)}dk^3\right)  +  O(|\Lambda|^{-1/3}) \nonumber \\
=  &&\frac{1}{3\pi^2}g^{3/2}_0  +  O(\rho^{\eta}) +  O(|\Lambda|^{-1/3})
\eeqa
This proves \eqref{4.19} 
since $L =\rho^{-25/24}$.

Similarly, for $u \in P_{I}$, we have 
\beqa
&&\lim_{\rho\to 0}|\Lambda|^{-1}\rho^{-3/2}\left(\sum_{u\in P_I}\frac{(\rho\lambda_u)^2}{1-(\rho\lambda_u)^2}\right)\\\nonumber
\le   &&\lim_{\rho\to 0}\frac{1}{(2\pi)^3}
\int_{\eta_L^{-1} \leq |k| \le \infty}
\frac{(h(k)-1)^2}{4h(k)}dk^3 +  O(|\Lambda|^{-1/3}) 
\eeqa
This proves \eqref{temp4.17} and concludes Proposition \ref{particlebound}.

\end{proof}

As a corollary to the proof,  we have the following estimates.
\begin{cor}
\beq\label{tail}
\lim_{ n \to \infty}\lim_{\rho\to 0}|\Lambda|^{-1}\rho^{-3/2}\left(\sum_{u\in P_L}\sum_{m=0}^{n}\left(\rho\lambda_u\right)^{2m}\right)
= \frac{g_0^{3/2}}{3\pi^2}
\eeq
\end{cor}
\begin{proof}
{F}rom the previous proof, we only need to prove  the tail terms vanishes. 
Recall $\left|\rho \lambda_u \right|\le 1 - \const \e_L < 1$ in \eqref{ublambda}. Thus we have 
\beqa\label{temp6.35}
&&\lim_{n\to \infty}\lim_{\rho\to 0}|\Lambda|^{-1}\rho^{-3/2}\left(\sum_{u\in P_L}\sum_{m=n+1}^{\infty}\left(\rho\lambda_u\right)^{2m}\right)
\\\nonumber
\le &&\lim_{n\to \infty}\lim_{\rho\to 0}|\Lambda|^{-1}\rho^{-3/2}    \left(\sum_{u\in \Lambda^*}  \frac{(\rho\lambda_u)^{2n+2}}{1-(\rho\lambda_u)^2}\right) \\\nonumber
\le &&\lim_{\rho\to 0}\frac{1}{(2\pi)^3}
\int_{\R^3}
H(2n)dk^3 +  O(|\Lambda|^{-1/3}) ,
\eeqa
where 
$$
H(2n)=\frac{(h(k)-1)^2\left(\frac{1-h(k)}{1+h(k)}\right)^{2n}}{4h(k)}.
$$
By Lebesgue monotone convergence theorem, we have that  $H(2n)$ converges to zero. This proves the 
Corollary. 
\end{proof}

\bigskip

We note that (\ref{temp6.29}) also shows that, for $u\in P_H$, $ Q_\Psi(\{u,m\})$ is exponentially small with $m$, i.e.,
\beq\label{temp5.19}
 Q_\Psi(\{u,m\})\leq (\const \left|\lambda_u\right|\rho^2|u|^{-2})^m.
\eeq
Furthermore, using similar method, one can easily generalize this result to: for $u,v\in P_H$ and $u+v\neq 0$ 
\beq\label{boundpsiuvmn}
 Q_\Psi(\{u,m\},\{v,n\})\leq (\const |\lambda_u|\rho^2\eps_H^{-2})^m(\const |\lambda_v|\rho^2\eps_H^{-2})^n, 
\eeq
which implies, for  $u,v\in P_H$ and $u+v\neq 0$, the following inequality:
\beq\label{boundpsifixeduv}
 Q_\Psi(u,v)\leq \const \left|\lambda_u\lambda_v\right|\rho^4\eps_H^{-4}.
\eeq

\subsection{Proof of Lemma \ref{boundspsi0}: Upper Bound}

Proposition \ref{particlebound} states that  the density  of particles with momenta in $P_{I}$ and $P_{H}$ are much smaller than $\rho^{3/2}$. And it  implies  
an upper bound on the density  of particles with momenta in $P_{L}$.  We now prove a matching lower bound
\beq\label{upt01}
\sum_{u\in P_L} Q_\Psi(u)\geq \left(\frac{1}{3\pi^2}g^{3/2}_0-\eps\right)\rho^{3/2}\Lambda
\eeq
for $\rho$  small enough. Since the total number of particles is fixed, this will provide a upper bound on the number of particles in the condensate and hence proves the upper bound part of Lemma \ref{boundspsi0}.

 We start with the following lemma, which  bounds the average number of particles in the condensate under the condition that there are at most $k$ particles with momentum $u$. 
\begin{prop}\label{prop6} For $u \in P_{I}$ and for any $k$ fixed with $0\leq k\leq m_c$ ($m_{c}$ defined in \eqref{defmc}), 
we have, for $\rho$ small enough,
\beqa\label{temp6.42}
\frac{\sum_{i=0}^kQ_\Psi(0 | \{u,i\}) Q_\Psi(\{u,i\})}{\sum_{i=0}^k Q_\Psi(\{u,i\})}\geq N-\const N\rho^{1/2}m_c .
\eeqa
\end{prop}

\begin{proof}  By \eqref{ubpsiuroughPL1},  the contribution of $\alpha \in M^a_u$ to $Q_\Psi(\{u,m\})$ for  $1\leq m\leq m_c$
is of lower order when compared with the contribution of $\alpha \in M^s_u$. The ratio of the contributions from $\alpha \in M^s_u$
between  $Q_\Psi(\{u,m\}) $ and $ Q_\Psi(\{u,m-1\})$ is estimated in \eqref{pubrho3}. Together with  the upper bound on  
$|\lambda_u|$ in (\ref{ublambda}) and the choices of  $\eps_L, \eps_H$, we have for $\rho$ small enough, 
\beq\label{mcmlbum0}
\frac{ Q_\Psi(\{u,m\})}{ Q_\Psi(\{u,m-1\})}\leq (\rho^2\lambda_u^2)(1+\const \frac{m_c\rho}{ \eps_H})\leq (1-\const (\eps_L- \frac{m_c\rho}{ \eps_H}))< 1.
\eeq
Hence $Q_\Psi(\{u,m\}) $ is monotonic decrease in $m$.  We thus  have for $0\leq k\leq m_c$,
\beq\label{mcmlbum1}
\sum_{i=0}^{k}Q_\Psi(\{u,i\}) 
\ge \frac{k+1}{m_c+1} \sum_{i=0}^{m_c}Q_\Psi(\{u,i\})  =  \frac{k+1}{m_c+1}, 
\eeq
where the last identity is the normalization of the state $\Psi$. 

By definition of  $Q_\Psi(0 | \{u,i\})$ and \eqref{resultlbpsi0}, we have 
$$
\sum_{i=0}^{m_c}Q_\Psi(0|\{u,i\}) Q_\Psi(\{u,i\})= 
Q_\Psi(0)\geq N-\const N\rho^{1/2}
$$
On the other hand, for any $m$, $Q_\Psi(0|\{u,m\})\leq N$. Hence, the numerator 
on the left side of  \eqref{temp6.42} can be bounded by:
\beqa\nonumber
   &&  \sum_{i=0}^kQ_\Psi(0 | \{u,i\}) Q_\Psi(\{u,i\}) \\\nonumber
    && = \sum_{i=0}^{m_c}Q_\Psi(0| \{u,i\}) Q_\Psi(\{u,i\})-\sum_{i=k+1}^{m_c}Q_\Psi(0 | \{u,i\}) Q_\Psi(\{u,i\})\\\label{temp6.45}
 && \geq N-\const N\rho^{1/2}-N\sum_{i=k+1}^{m_c}   Q_\Psi(\{u,i\}) \\\nonumber
&&  =N{\sum_{i=0}^{k} Q_\Psi(\{u,i\})}-\const N\rho^{1/2}, \\\nonumber
\eeqa
where we have used  $\sum_{i=0}^{m_c} Q_\Psi(\{u,i\})=1$ in the last identity. 
Finally, we  divide \eqref{temp6.45} by $\sum_{i=0}^{k} Q_\Psi(\{u,i\})$ and use  \eqref{mcmlbum1} to conclude 
\eqref{temp6.42}. 
\end{proof}

Return to the proof of \eqref{upt01} for $u\in P_L$. Since $\mathcal A^{u}\beta$ is a one to one map 
(not necessarily surjective), we have 
\beq\label{lastlb2.1}
\sum_{i=1}^{m_c} Q_\Psi(\{u,i\}) \ge \sum_{\be(u)=0}^{m_c-1}
|f(\mathcal A^{u}\beta)|^2 
\eeq
{F}rom \eqref{properf2} and \eqref{properf3}, the right hand side is bounded below by 
\beq\label{lastlb2}
\lambda_u^2|\Lambda|^{-2}\sum_{\be(u)=0}^{m_c-1}
\left(\be(0)^2-\be(0)\right)|f(\be)|^2
\eeq
By Jensen's inequality and $\be(0)\leq N$, it is bounded below by 
\beq\label{lastlb3}
\lambda_u^2|\Lambda|^{-2}\left(\left(\frac{\sum_{\be(u)=0}^{m_c-1}
\be(0)|f(\be)|^2}{\sum_{\be(u)=0}^{m_c-1}
|f(\be)|^2}\right)^2-N\right)\sum_{\be(u)=0}^{m_c-1}
|f(\be)|^2
\eeq
By definition, 
\beq\label{temp6.49}
\frac{\sum_{\be(u)=0}^{m_c-1}
\be(0)|f(\be)|^2}{\sum_{\be(u)=0}^{m_c-1}
|f(\be)|^2}=\frac{\sum_{i=0}^{m_c-1}Q_\Psi(0 | \{u,i\}) Q_\Psi(\{u,i\})}{\sum_{i=0}^{m_c-1} Q_\Psi(\{u,i\})}
\eeq
The term on the right hand side can be estimated by  Proposition \ref{prop6}. Combining all estimates 
up to now and  we obtain 
\beq
\sum_{i=1}^{m_c} Q_\Psi(\{u,i\})\geq ((\rho-\rho^{5/4})\lambda_u)^{2}\sum_{i=0}^{m_c-1} Q_\Psi(\{u,i\})
\eeq
Finally, using \eqref{mcmlbum1}, we have
\beq
\sum_{i=1}^{m_c} Q_\Psi(\{u,i\})\geq  ((\rho-\rho^{5/4})\lambda_u)^{2}\left(1-\frac{1}{m_c+1}\right)
\eeq

We can generalize this result as follows. For $m\geq 1$,  we first iterate the argument in proving  \eqref{lastlb2.1}
and \eqref{lastlb2} to have 
\beqa\label{lastlb4}
&&\lambda_u^{2m}|\Lambda|^{-2m}\sum_{\be(u)=0}^{m_c-m}
\left(\be(0)-2m\right)^{2m}|f(\be)|^2 \le \sum_{i=m}^{m_c} Q_\Psi(\{u,i\})
\eeqa
Again, using Jensen's inequality, Proposition \ref{prop6}, and \eqref{mcmlbum1}, we have 
\beq
\sum_{i=m}^{m_c} Q_\Psi(\{u,i\}\geq  ((\rho-\rho^{5/4})\lambda_u)^{2m}\left(1-\frac{m}{m_c+1}\right)
\eeq
So with the fact $m_c=\rho^{-\eta}$, $Q_\Psi(u)$ can be bounded as follows, 
\beqa\nonumber
 Q_\Psi(u)= \sum_{m=1}^\infty\sum_{i=m}^\infty  Q_\Psi(\{u,i\})
\geq&& \sum_{ m=1}^{m_c}((\rho-\rho^{5/4})\lambda_u)^{2m}\left(1-\frac{m}{m_c+1}\right)\\\label{ubpsiuPL}
\ge&&(1-\rho^{\eta/2})  \sum_{i=1}^{\sqrt{m_c}+1}(\rho\lambda_u)^{2i}\\\nonumber
\eeqa
Now the summation  over $u \in P_{L}$ was carried out in Corollary \ref{tail}
and  we have proved   (\ref{upt01}).  Since the total number of particle is $N$, 
the  bounds on $Q_\Psi(0)$ follows from (\ref{upt01}) and Proposition \ref{particlebound}. 
This concludes Lemma \ref{boundspsi0}. 

\bigskip 

 
The previous  method can be applied to yield the following estimates which will be useful later on. 
\begin{lem} For $u\in P_L$ and  $\rho$ sufficiently small, the following two bounds hold: 
\beq\label{ubpsiuPLmc}
	\sum_{m=m_c-1}^{m_c} m Q_\Psi(\{u,m\})\leq \frac{\rho^2\lambda_u^2}{1-\rho^2\lambda_u^2}\,\rho^{\eta/2}
\eeq
\beq\label{lbpsiuPL0}
	\sum_{\al(u)\leq m_c-2}|f(\al)|^2\al(0)^2\al(u)\geq N^2\frac{\rho^2\lambda_u^2}{1-\rho^2\lambda_u^2}(1-2\rho^{\eta/2}-(\rho\lambda_u)^{2\sqrt{m_c}}).
\eeq 
 \end{lem}
 \begin{proof}
Because $Q_\Psi(\{u,m\}) $ is monotonic decrease in $m$, we have 
 \beq
 \sum_{m=m_c-1}^{m_c} m Q_\Psi(\{u,m\})\leq \frac{\const} {m_c} \sum_{m=1}^{m_c} m Q_\Psi(\{u,m\})=  \frac{\const} {m_c}Q_{\Psi}(u)
 \eeq
Together with the upper bound \eqref{ubpsiuroughPM} on $Q_\Psi(u)$, we have proved \eqref{ubpsiuPLmc}.

To prove \eqref{lbpsiuPL0}, we follow the argument in \eqref{lastlb4} to  have, for $m\leq m_{c}-2$,
\beq\nonumber
\lambda_u^{2m}|\Lambda|^{-2m}\sum_{\be(u)=0}^{m_c-m-2}
\left(\be(0)-2m\right)^{2m+2}|f(\be)|^2 \le \sum_{i=m}^{m_c-2} Q_\Psi(0,0| \{u,i\})Q_\Psi(\{u,i\})
\eeq
Again, using Jensen's inequality, Proposition  \ref{prop6} and \eqref{mcmlbum1}, we have 
\beqa\nonumber
\sum_{\al(u)\leq m_c-2}|f(\al)|^2\al(0)^2\al(u)&&=\sum_{m=0}^{m_c-2}\sum_{i=m}^{m_c-2} Q_\Psi(0,0 | \{u,i\})Q_\Psi(\{u,i\})\\
&&\geq (1-2\rho^{\eta/2})  \sum_{i=1}^{\sqrt{m_c}}(\rho\lambda_u)^{2i}N^2
\eeqa
This implies \eqref{lbpsiuPL0}.
 \end{proof}

Lemma \ref{boundspsi0} can be extended to the following estimate: 

\begin{lem}\label{boundspsi02}
With the assumptions in Lemma \ref{boundspsi0},  $Q_\Psi(0,0)$ satisfies the estimate
\beqa
\big(\Lambda\rhom \big)^2\leq Q_\Psi\big(0, 0 \big)\leq \big(\Lambda\rhop\big)^2 
\eeqa
\end{lem}

\begin{proof}
By  Jensen's inequality and Lemma \ref{boundspsi0},   we have the lower bound
\[
Q_\Psi(0,0)\geq \left[Q_\Psi(0)\right]^2\geq \big(\Lambda\rhom \big)^2
\]

For the upper bound,  we start with  
\beqa\label{temp6.56}
Q_\Psi(0,0)
&= & N^2-2N\sum_{u\neq 0} Q_\Psi(u)+\sum_{u,v\neq 0} Q_\Psi(u,v)\\\nonumber
&\leq & (Q_\Psi(0))^2+\sum_{u,v\neq 0} Q_\Psi(u,v)
\eeqa
Since the number of particles with momentum $u\in P_L$ is at most  $m_c$, 
\beqa\label{6.6}
\sum_{u\in P_L,v\neq 0} Q_\Psi(u,v) &&\leq\left(\sum_{u\in P_L}m_c\right)\left(\sum_{v\neq0}Q_\Psi(v)\right)
\eeqa
By definition of $P_L$, we have $\sum_{u\in P_L}m_c=m_c\eta_L^{-3}\rho^{3/2}\Lambda$.  
The last factor in \eqref{6.6} can be estimated by Proposition \ref{particlebound}. Thus we have
\beq\label{temp6.58}
\sum_{u\in P_L,v\neq 0} Q_\Psi(u,v)=o(\rho^{5/2}|\Lambda|^2)
\eeq
For the terms $\sum_{u\in P_I\cup P_H,v\neq 0}$, 
the upper bound on the total number of particles in $P_I$ and $P_H$ in Proposition \ref{particlebound} yields that 
\beqa\label{temp6.59}
\sum_{u\in P_I\cup P_H,v\neq 0} Q_\Psi(u,v)&&\leq\sum_{u\in P_I\cup P_H} Q_\Psi(u) N =o(\rho^{5/2}|\Lambda|^2)
\eeqa
Inserting \eqref{temp6.58}, \eqref{temp6.59} into  \eqref{temp6.56}
and using the upper bound in Lemma \ref{boundspsi0},  we obtain the upper bound on $Q_\Psi(0,0)$. 
\end{proof}

\section{Estimates on  Kinetic Energy}

In this section, we will prove the kinetic energy estimate Lemma \ref{lemkinetic}.
This lemma follows immediately from summing the estimates (\eqref{boundspsiu2PMtotal}-\eqref{boundspsiu2PHtotal})of the next lemma. 

\begin{lem}\label{boundspsiu} 
In the limit $\rho\to0$, $Q_\Psi(u,\,v)$ can be bounded  above by 
\beq\label{boundspsiuvtotal}
\lim_{\rho\to 0} \left(\rho^{-5/2}|\Lambda|^{-2}\sum_{u,v\neq 0} Q_\Psi(u,v)\right) \le 0
\eeq
Furthermore, $\sum u^2  Q_\Psi(u)$ can be bounded  above as follows
\beqa\label{boundspsiu2PMtotal}
&&\overline\lim_{\rho\to 0} \left(\rho^{-5/2}|\Lambda|^{-1}\sum_{u\in P_I}u^2\left( Q_\Psi(u)-(\rho_0\lambda_u
)^2   \right)\right)\le 0\\
\label{boundspsiu2PLtotal}
&&\overline\lim_{\rho\to 0} \left(\rho^{-5/2}|\Lambda|^{-1}\sum_{u\in P_L}u^2\left( Q_\Psi(u)-(\rho_0w_u)^2\right)\right)\leq-\frac{8}{5\pi^2}g_0^{5/2}\\
\label{boundspsiu2PHtotal}
&&\overline\lim_{\rho\to 0} \left(\frac{\rho^{-5/2}}{|\Lambda|}\!\!\sum_{u\in P_H}u^2\!\!\left( Q_\Psi(u)\!-\!\left(\rho_0^2+\frac{4g^{3/2}_0}{3\pi^2}\rho_0^{5/2}\right)\!\lambda^2_u   \!\right)\!\!\right) \le 0
\eeqa
\end{lem}

%

\bigskip

\begin{proof} The bound  (\ref{boundspsiuvtotal}) was proved in \eqref{temp6.58} and \eqref{temp6.59}. 
We now prove  (\ref{boundspsiu2PMtotal}) concerning $ u\in P_I$.

The upper bound of $ Q_\Psi(u)$   in (\ref{ubpsiuroughPM}) can be rewritten as  
\beq
 Q_\Psi(u) \leq (\rho\lambda_u)^2+\frac{(\rho\lambda_u)^4}{1-(\rho\lambda_u)^2}
\eeq
Recall $\rho_{0} = \rho(1 + O(\sqrt \rho))$  and 
the bounds on $\lambda$ in (\ref{ublambda}).  
Since $\rho^{1/2}\ll|u|\ll 1$ when $u\in P_I$, see Definition \ref{def1}, the error term of the last bound can be estimated by 
\beqa
&&\overline\lim_{\rho\to0} |\Lambda|^{-1} \rho^{-5/2} \sum_{u: \rho^{1/2}\ll|u|\ll 1}u^2\frac{(\rho\lambda_u)^4}{1-(\rho\lambda_u)^2}=0
\eeqa
This proves  (\ref{boundspsiu2PMtotal}).

\par We now prove  \eqref{boundspsiu2PLtotal} concerning $ u\in P_L$.
Following the strategy of the previous argument, we first  use $0\ge 1-(\rho_0\lambda_u)^2 \ge \const \e_{L}$ 
in \eqref{ublambda} and \eqref{lublambdaPL0} to 
rewrite  the upper bound of $ Q_\Psi(u)$ in (\ref{ubpsiuroughPL}) as
\beq
 Q_\Psi(u)
\leq \frac{(\rho\lambda_u)^2}{1-(\rho\lambda_u)^2}+\const\frac{\rho m_c}{\eps_H\e_L}
\eeq
The error terms are negligible in the sense that 
\[
\sum_{u\in P_L}u^2\frac{\rho m_c}{\eps_H\e_L}=o(\rho^{5/2}\Lambda)
\]
Since  $w_u=g_u |u|^{-2}$, $\rho_0-\rho=O(\rho^{3/2})$ and $|g_u-g_0|\leq \const |u|$,   we have
\beq
\lim_{\rho}\sum_{u\in P_L}u^2\left( (\frac{\rho g_0}{u^2} )^2-(\rho_0w_u)^2\right)\rho^{-5/2}|\Lambda|^{-1}=0
\eeq
Summarize what we have proved, we have the following  inequality:
\beqa\label{temp5.40}
&&\overline\lim_{\rho}\sum_{u\in P_L}u^2\left( Q_\Psi(u)-(\rho_0w_u)^2\right)\rho^{-5/2}|\Lambda|^{-1}\\\nonumber
\leq &&\overline\lim_{\rho}\sum_{u\in P_L}u^2\left(\frac{(\rho\lambda_u)^2}{1-(\rho\lambda_u)^2}-(\frac{\rho g_0}{u^2} )^2\right)\rho^{-5/2}|\Lambda|^{-1}
\eeqa
Let $u=\sqrt{\rho}k$ and $h(k)=\sqrt{1+4g_0|k|^{-2}}$ as in \eqref{defhk}. Then the right hand side  of \eqref{temp5.40} is estimated as 
\beq\nonumber
 \frac{1}{(2\pi)^3}\int_{\eps_L\leq |k|\leq\eta_L^{-1}}k^2\left(\frac{1+2g_0|k|^{-2}}{2h(k)}-\frac{1+2(g_0|k|^{-2})^2}2\right)dk^3+O(|\Lambda|^{-1/3})
\eeq
Direct calculation yields that 
\beq
\frac{1}{(2\pi)^3}\int_{k\in \R^3}k^2\left(\frac{1+2g_0|k|^{-2}}{2h(k)}-\frac{1+2(g_0|k|^{-2})^2}2\right)dk^3=-\frac{8}{5\pi^2}g^{5/2}_0,
\eeq
Inserting this result into (\ref{temp5.40}), 
we obtain  the desired result (\ref{boundspsiu2PLtotal}).

Finally, we  prove  \eqref{boundspsiu2PHtotal} concerning $ u\in P_H$.
Recall the bound  (\ref{temp6.29}) on the ratio of $ Q_\Psi(\{u,m\})/ Q_\Psi(\{u,m-1\})$. Since  
$|\lambda_u|\leq g_0 |u|^{-2}$ (\ref{ublambda}) and $u \in P_{H}$, the factor on the right hand side  of
(\ref{temp6.29}) can be bounded by $\rho^{3/2}$.
Thus we have 
\beq\label{temp7.16}
 Q_\Psi(u)=\sum_{m}m  Q_\Psi(\{u,m\})\leq
\sum_{m\geq 1}  Q_\Psi(\{u,m\})(1+O(\rho^{3/2}))
\eeq
We now repeat the argument from \eqref{4.1} to \eqref{temp6.29} but refine the proof by using Proposition 
\ref{particlebound}. 
Hence for any $u\in P_H$, we have 
\beqa\nonumber
&&\sum_{m\geq 1}  Q_\Psi(\{u,m\})\\\nonumber
&\leq&\sum_{\be} \left(\frac{\be(0)^2}{|\Lambda|^2}\lambda_u^2+\sum_{v\in P_L}4\frac{\be(0)}{|\Lambda|}\frac{\be(v)}{|\Lambda|}\left|\lambda_{u}\lambda_{-u+v}\right|\right)|f(\be)|^2\\\label{proofubpsiu12}
&\leq& |\Lambda|^{-2}\lambda_u^2Q_\Psi(0,0)+\sum_{v\in P_L}\rho |\Lambda|^{-1}\left(4Q_\Psi(v)\left|\lambda_{u}\lambda_{-u+v}\right|\right)
\eeqa
By mean value theorem and $\lambda_k=-g_k|k|^{-2}$ for $k\in P_H$, we have that  $\exists\tilde u\in \R^3: |\tilde u-u|\leq v$ s.t. 
\beqa
|\lambda_{-u+v}-\lambda_{-u}|&\leq &\const \left(\left|\frac{\partial g_{\tilde u}}{\partial {\tilde u}}\right|{\tilde u}^{-2}+|g_{\tilde u}|{\tilde u}^{-3}\right)|v|
\eeqa
{F}rom  the estimates \eqref{ublambda} on $\lambda_u$ and $u\sim \tilde u$, we obtain:
\beqa
|\lambda_u||\lambda_{-u+v}-\lambda_{-u}|&\leq &\const \left(\left|\frac{\partial g_{\tilde u}}{\partial {\tilde u}}g_u\right|u^{-4}+|g_{\tilde u}||g_u|u^{-5}\right)|v|\nonumber\\
&\leq &\const |u|^{-2}\eps_H^{-3}G(u)|v|,
\eeqa 
where by Schwarz inequality, we have:
\beq
G(u)=\max_{u':|u'-u|\leq\eta_L^{-1}\rho^{1/2}}\left\{\left|\frac{\partial g_{ u'}}{\partial {u'}}\right|^2+|g_{u'}|^2\right\}
\eeq
We note that it is easy to check $\sum_{u\in P_H}G(u)/\Lambda<\infty$. Together with the results on the total number of $P_L$ particles in (\ref{boundspsiuPLtotal}), 
we obtain that, for $\rho$ small enough  and $u\in P_H$, the last term in (\ref{proofubpsiu12}) is bounded above by 
\beqa\label{temp7.20}
 \lambda_u^2\rho^{5/2}\left(\frac{4g^{3/2}_0}{3\pi^2}+\rho^\eta\right)+ \frac{\const\rho^{3}}{u^2\eps_H^{3}\eta_L}G(u)
\eeqa
The $Q_\psi(0,0)$ in the last second term of \eqref{proofubpsiu12} is bounded by Lemma \ref{boundspsi02}. Inserting (\ref{temp7.20}) and \eqref{proofubpsiu12} into (\ref{temp7.16}) and using  $\lambda_u^2=w_u^2$ for $u\in P_H$, we obtain that,
\beq
\overline\lim_{\rho\to0}\sum_{u\in P_H}u^2\left( Q_\Psi(u)\!-\!(\rho_0w_u)^2\left(1+\left[\frac{4g^{3/2}_0}{3\pi^2}\right]\rho_0^{1/2}\right)\!\right)\frac{\rho^{-5/2}}{|\Lambda|}\leq 0
\eeq
This proves  (\ref{boundspsiu2PHtotal}). 
\end{proof}

\section{Estimates on Pair Interaction Energies}

\subsection{Proof of Lemma \ref{lemHS1}}
First, with the fact $a_u^\dagger a_u^\dagger a_ua_u\leq (a_u^\dagger a_u)^2$ and $0\leq |V_u|\leq V_0$ for any $u$, we can bound $H_{S1}$ as follows
\beqa
H_{S1}\leq&& V_0\Lambda^{-1}\sum_{u,v} a_u^\dagger a_u a_v^\dagger a_v +\Lambda^{-1}\sum_{u\neq v} V_{u-v}a_u^\dagger a_u a_v^\dagger a_v\\\nonumber
\leq &&V_0N\rho+V_0\Lambda^{-1}\sum_{u\neq v}a_u^\dagger a_v^\dagger a_va_u=2V_0N\rho-V_0\Lambda^{-1}\sum_{u}(a_u^\dagger a_u)^2
\eeqa
Therefore we can bound the expectation value $\langle H_{S1}\rangle$:
\beq
\langle H_{S1}\rangle_\Psi\leq 2V_0N\rho-V_0\Lambda^{-1}\sum_{u}Q_\Psi(u,u)\leq 2V_0N\rho-V_0\Lambda^{-1}Q_\Psi(0,0)
\eeq
By the lower bounds of $Q_\Psi(0,0)$ in Lem. \ref{boundspsi02} and the definition of $\rho_0$ in \eqref{defrho0}, we have proved Lemma \ref{lemHS1}.

\subsection {Proof of Lemma \ref{lemHS2}}


We start the proof with  the following identity for $\langle\Psi|a^\dagger_{u_1}a^\dagger_{u_2}a_{u_3} a_{u_4}|\Psi\rangle$. 
\begin{lem}\label{id4a}
For any fixed $u_{1,2,3,4}\in \Lambda^*$ and $\al\in M$,  define  $T(\al)$ to be the state 
\beq
|T(\al)\rangle=Ca^\dagger_{u_1}a^\dagger_{u_2}a_{u_3} a_{u_4}|\al\rangle, 
\eeq 
where $C$ is the positive normalization constant when  $| T(\al)\rangle\not =0$. Then we have
\beq\label{resultid4a}
\langle\Psi|a^\dagger_{u_1}a^\dagger_{u_2}a_{u_3} a_{u_4}|\Psi\rangle=\sum_{\al\in M}f(\al)\overline {f(T(\al))}
\sqrt{\langle\al|a^\dagger_{u_4}a^\dagger_{u_3}a_{u_2}a_{u_1} |a^\dagger_{u_1}a^\dagger_{u_2}a_{u_3} a_{u_4}|\al\rangle}
\eeq
\end{lem}
The map $T$ depends on $u_{1,2,3,4}$ and in principle it has to carry them as subscripts. We omit these subscripts since it will be clear
from the context  what they are. 

\begin{proof} 
For any $u_{1,2,3,4}\in \Lambda^*$ fixed, by definition of $\Psi$, we have
\beq\label{proofid4a1}
\langle\Psi|a^\dagger_{u_1}a^\dagger_{u_2}a_{u_3} a_{u_4}|\Psi\rangle=\sum_{\al,\be\in M}f(\al) {\overline{f(\be)}}\langle\be|a^\dagger_{u_1}a^\dagger_{u_2}a_{u_3} a_{u_4}|\al\rangle
\eeq
By definition of $M$, we have 
\beq\label{proofid4a2}
\langle\be|a^\dagger_{u_1}a^\dagger_{u_2}a_{u_3} a_{u_4}|\al\rangle\neq 0\Rightarrow \be=T(\al)
\eeq
Since  $|T(\al)\rangle$ is normalized, the identity in Lemma \ref{id4a} is obvious. 
\end{proof}

Lemma \ref{lemHS2} follows from the following lemma and $\lambda_u=-w_u$ for $u\in P_H\cup P_I$.
Notice that the factor $2$ in the estimate of Lemma \ref{lemHS2} is due to the complex conjugate in the definition of 
$H_{S2}$. Similar factor also appears in Lemma \ref{lemHAS1}. 
\begin{lem}\label{symno20}
\par 
\beqa\label{u-uPMH}
\overline\lim_{m_c,\,\rho}\sum_{u\in P_I\cup P_H}\left(\langle V_u|\Lambda|^{-1}a^\dagger_ua^\dagger_{-u}a_0a_0\rangle- \rho_0^{2}V_u  \lambda_u\right)\rho^{-5/2}|\Lambda|^{-1}=0
\\\label{u-uPL}
\overline\lim_{m_c,\,\rho}\sum_{u\in P_L}\left(\langle V_u|\Lambda|^{-1}a^\dagger_ua^\dagger_{-u}a_0a_0\rangle+\rho_0^2V_uw_u\right)\rho^{-5/2}|\Lambda|^{-1}\leq \frac{V_0g_0^{3/2}}{\pi^2}
\eeqa
\end{lem}

\begin{proof}
We first prove \eqref{u-uPMH} concerning with $u\in P_I\cup P_H$. By Lemma \ref{id4a}, we have
\beqa\label{prle81}
&&\langle V_u|\Lambda|^{-1}a^\dagger_ua^\dagger_{-u}a_0a_0\rangle\\\nonumber
=&&V_u|\Lambda|^{-1}\!\!\!\!\!\!\sum_{\al:\al\in M,  \mathcal A^u \al \in M }\!\!\!f(\al)f( \mathcal A^u \al )\sqrt{(\al(0)^2-\al(0))(\al(u)+1)(\al(-u)+1)}
\eeqa
The case that $\al\in M$ and $ \mathcal A^u \al \notin M$ can only happen when $\al(0)=0$ or $1$
and thus has no contribution. 
{F}rom the relation between $f(\al)$ and $f( \mathcal A^u \al )$ in \eqref{properf1},  we have
\beq\label{temp8.15}
\eqref{prle81}=\lambda_uV_u|\Lambda|^{-2}\sum_{\al\in M }|f(\al)|^2\al(0)(\al(0)-1)\sqrt{(\al(u)+1)(\al(-u)+1)}
\eeq
By the Schwarz inequality, we have  
\beqa\label{temp8.16}
&&\left|\sum_{\al}\al(0)(\al(0)-1)\left(\sqrt{(\al(u)+1)(\al(-u)+1)}-1\right)|f(\al)|^2\right|\nonumber\\
\leq&&N^2\left|\sum_{\al}\frac{\al(u)+\al(-u)}2|f(\al)|^2\right|=N^2 Q_\Psi(u)
\eeqa
Inserting \eqref{temp8.16} into \eqref{temp8.15} and  summing  over $u\in P_I\cup P_H$ of \eqref{temp8.15},  we obtain 
\beqa
&&\sum_{u\in P_I\cup P_H}\left(\langle V_u|\Lambda|^{-1}a^\dagger_ua^\dagger_{-u}a_0a_0\rangle-V_u\lambda_u(Q_\Psi(0,0)-Q_\Psi(0))\right)\\\nonumber
\leq&&\const \rho^2|\Lambda|\sum_{u\in P_I\cup P_H}Q_\Psi(u)
\eeqa
{F}rom  the upper bound of $\sum Q_\Psi(u)$ in \eqref{temp4.17},  the right hand side  of above inequality is 
bounded by ($o(\rho^{5/2}\Lambda)$).  By  the bounds on $Q_\Psi(0,0)$ in  Lemma \ref{boundspsi02}, we have proved  \eqref{u-uPMH}.

To prove \eqref{u-uPL} concerning  $u\in P_L$, we note that  \eqref{prle81} still holds, but
$ \mathcal A^u \al \notin M$ when $\al^*(u)=m_c$. Therefore, for $u\in P_L$, \eqref{prle81} is equal to
\beq\nonumber
V_u|\Lambda|^{-1}\sum_{\al:\al\in M, \al^*(u)<m_c}f(\al)f( \mathcal A^u \al )\sqrt{\al(0)(\al(0)-1)(\al(u)+1)(\al(-u)+1)}
\eeq
We can express $f( \mathcal A^u \al )$ in terms of $f(\al )$; in both cases: $\al\in M_u^s$ or $\al\in M_u^a$, we have the following identity: 
\beqa
&& f(\al)f( \mathcal A^u \al )\sqrt{(\al(u)+1)(\al(-u)+1)} \\&=& \lambda_u |f(\al)|^2|\Lambda|^{-1} \sqrt{\al(0)(\al(0)-1)}(\al^*(u)+1)
\eeqa
Hence, for $u\in P_L$, 
\beq\label{temp8.23}
\eqref{prle81}= \sum_{\al:\al\in M, \al^*(u)<m_c}\lambda_uV_u|\Lambda|^{-2}|f(\al)|^2\al(0)(\al(0)-1)(\al^*(u)+1)
\eeq
We note $\lambda_u<0$ and  $V_u\approx V_0>0$, for $u\in P_L$. For any $\alpha \in M$, $\al^*(u)-\al(u)\leq 1$ by definition.    Hence we can replace the summation $\al^*(u)<m_c$ 
by  $\al(u)\leq m_c-2$ to have an upper bound. Summing over $u\in P_L$ of \eqref{temp8.23}, we have
\beqa\nonumber
&&\langle \sum_{u\in P_L}V_u|\Lambda|^{-1}a^\dagger_ua^\dagger_{-u}a_0a_0\rangle\\\nonumber
 &\leq& \sum_{u\in P_L}
\sum_{\al(u)\leq m_c-2}\lambda_uV_u|\Lambda|^{-2}|f(\al)|^2\al(0)(\al(0)-1)\al(u)\\\label{temp8.24}
&& +\sum_{u\in P_L} \sum_{\al(u)\leq m_c-2}\lambda_uV_u|\Lambda|^{-2}|f(\al)|^2\al(0)(\al(0)-1)
\eeqa
The last term is equal to 
\beqa\label{temp8.25}
& &\sum_{u\in P_L} \lambda_uV_u|\Lambda|^{-2}(Q_\Psi(0,0)-Q_\Psi(0))\\\nonumber
&-&\sum_{u\in P_L} \sum_{i=m_c-1}^{m_c}\lambda_uV_u|\Lambda|^{-2}(Q_\Psi(0,0|{u,i})Q_\Psi({u,i}))
\eeqa
Since  $Q_\Psi(0,0|{u,i})\leq N^{2}$,  the last term in \eqref{temp8.25} is bounded from above by
\beq
\sum_{u\in P_L} \sum_{i=m_c-1}^{m_c}\const|\lambda_u\rho^2 Q_\Psi({u,i}))| \le o(\rho^{5/2}\Lambda),
\eeq
where we have used \eqref{mcmlbum0}. 
For the first term of \eqref{temp8.25}, we can bound it by using  Lemma \ref{boundspsi02}.
We now use \eqref{lbpsiuPL0} to estimate the first term on the right hand side of \eqref{temp8.24}. 
Combining these results, we have 
\beqa\nonumber
\langle \sum_{u\in P_L}V_u|\Lambda|^{-1}a^\dagger_ua^\dagger_{-u}a_0a_0\rangle\!\!\!&&\leq\!\sum_{u\in P_L}\!\lambda_uV_u\rho^2\frac{(\rho\lambda_u)^2}{1-(\rho\lambda_u)^2}(1-2\rho^{\frac\eta2}-(\rho\lambda_u)^{2\sqrt{m_c}})\\\label{temp8.26}
\!\!\!&&+\sum_{u\in P_L}\lambda_uV_u\rho_0^2+o(\rho^{5/2}\Lambda)
\eeqa
Since $\left|\lambda_u\rho\right|\leq 1$ and $|V_u|\leq V_0$, we have
\beq\label{temp8.27}
\sum_{u\in P_L}|\lambda_uV_u|\rho^2\frac{(\rho\lambda_u)^2}{1-(\rho\lambda_u)^2}\leq \sum_{u\in P_L}V_0\rho\frac{(\rho\lambda_u)^2}{1-(\rho\lambda_u)^2}\leq \const\rho^{5/2}\Lambda
\eeq
By  \eqref{temp6.35}, we have 
\beq\label{temp8.28}
\sum_{u\in P_L}|\lambda_uV_u|\rho^2\frac{(\rho\lambda_u)^2}{1-(\rho\lambda_u)^2}(\rho\lambda_u)^{2\sqrt{m_c}}\leq  o(\rho^{5/2}\Lambda)
\eeq
Inserting \eqref{temp8.27}-\eqref{temp8.28} into \eqref{temp8.26}, we have
\beqa\label{temp8.29}
&&\sum_{u\in P_L}\left(\langle V_u|\Lambda|^{-1}a^\dagger_ua^\dagger_{-u}a_0a_0\rangle+w_uV_u\rho_0^2\right)
\nonumber \\
&&  \le  \sum_{u\in P_L}(\lambda_u+w_u)V_u\rho^2_0 +\sum_{u\in P_L}V_u\rho^2\frac{\lambda_u^3\rho^2}{1-\rho^2\lambda_u^2}+o(\rho^{5/2}\Lambda)
\eeqa
Since  $|g_u-g_0|+ |V_{u}-V_{0}| \leq \const |u|$,  we can replace $w_u$ and $V_u$ by $g_0|u|^{-2}$ and $V_0$ in last inequality so that the rhs of \eqref{temp8.29} is bounded by 
\beqa
&&V_0\rho^2_0 \sum_{u\in P_L}(\lambda_u+g_0|u|^{-2}) +V_0\rho^2 \sum_{u\in P_L}\frac{\lambda_u^3\rho^2}{1-\rho^2\lambda_u^2}+o(\rho^{5/2}\Lambda)\\\nonumber
=&&V_0\rho^2_0 \sum_{u\in P_L}\left(\lambda_u+g_0|u|^{-2}+\frac{\lambda_u^3\rho^2}{1-\rho^2\lambda_u^2}\right)+o(\rho^{5/2}\Lambda)
\eeqa 
Let $u=\sqrt \rho k$. We have
\beqa
&&\lim_{\rho\to 0}\sum_{u\in P_L}\left(\lambda_u+g_0|u|^{-2}+\frac{\lambda_u^3\rho^2}{1-\rho^2\lambda_u^2}\right)\rho^{-1/2}|\Lambda|^{-1}\\\nonumber
=&&\lim_{\rho\to 0}\frac{1}{(2\pi)^3}\int_{\eps_L\leq |k|\leq \eta_L^{-1}}
 g_0|k|^{-2}\left(1-\frac{1}{\sqrt{1+4g_0|k|^{-2}}}\right)dk^3\\\nonumber
 =&&\pi^{-2}
\eeqa
So the leading term of right hand side of \eqref{temp8.29}   is equal to $V_0g_0^{3/2}\pi^{-2}(\rho^{5/2}\Lambda)$.  This completes the proof for \eqref{u-uPL}.
\end{proof}

\subsection{Proof of Lemma \ref{lemHS3}}

Define $P(u,v)$ by
\beq\label{defPuv}P(u,v)\equiv\sum_{\gamma\in M}f( \mathcal A^u\gamma)f( \mathcal A^v\gamma)
\sqrt{(\gamma(u)+1)(\gamma(-u)+1)(\gamma(v)+1)(\gamma(-v)+1)}
\eeq
Recall  $f(\al)=0$ when $|\al\rangle=0$ or $\al\notin M$.
\begin{lem}\label{id4a2}
Let $u$, $v\in \Lambda^*$, $u\neq v$ and  $u,v\neq 0$. 
\par If one of $u$ and $v$  $\in P_L\cup P_I$, we have the following identity.
\beqa\label{resultid4a2}
&&\langle\Psi|a_u^\dagger a_{-u}^\dagger a_v  a_{-v}|\Psi\rangle=P(u,v)
\eeqa
\par If $u,v\in P_H$, we have
\beqa\label{resultid4a2o}|\langle\Psi|a_u^\dagger a_{-u}^\dagger a_v  a_{-v}|\Psi\rangle-P(u,v)|\leq \const\rho^4\left|\lambda_u\lambda_v\right|
\eeqa
\end{lem}
\begin{proof}
We first  prove \eqref{resultid4a2} and assume without loss of generality  that $v\in P_L\cup P_I$. Using Lemma \ref{id4a}, we rewrite $\langle a_u^\dagger a_{-u}^\dagger a_v  a_{-v}\rangle_\Psi$ as
\beq\label{id4a21}
\langle a_u^\dagger a_{-u}^\dagger a_v  a_{-v}\rangle_\Psi=\sum_{\al\in M}f(\al)f(T(\al))\sqrt{(\al(u)+1)(\al(-u)+1)\al(v)\al(-v)}
\eeq
Here $|T(\al)\rangle=Ca_u^\dagger a_{-u}^\dagger a_v  a_{-v}|\al\rangle$ and $C$ is positive normalization constant. Since $v\in P_L\cup P_I$ and $\al(v)>0, \al(-v)>0$, by definition of $M$  
there exists unique $\gamma\in M$ such that 
\beq\label{gammaDaal}
 \mathcal A^v\gamma=\al
\eeq
Therefore, with $|T(\al)\rangle=Ca_u^\dagger a_{-u}^\dagger a_v  a_{-v}|\al\rangle$, we have
\beq\label{gammaDabe}T(\al)= \mathcal A^u\gamma.
\eeq
Furthermore, by \eqref{gammaDaal}, we have
\beq\label{gammaualu}
\gamma(u)=\al(u) \rmand \gamma(v)=\al(v)+1.
\eeq 
Inserting \eqref{gammaDaal}, \eqref{gammaDabe} and \eqref{gammaualu} into \eqref{id4a21}, we 
have proved  \eqref{resultid4a2}.

To  prove \eqref{resultid4a2o}, 
we define $N_v$ as the following set:
\beq
N_v\equiv\left\{\al\in M|\forall \gamma\in M,  \mathcal A^v\gamma\neq \al\right\}
\eeq
Following the previous argument, we have 
\beq
\left|\langle a_u^\dagger a_{-u}^\dagger a_v  a_{-v}\rangle_\Psi-P(u,v)\right|\leq\sum_{\al\in N_v,\be\in N_u} \left|f(\al)f(\be)\langle \be|a_u^\dagger a_{-u}^\dagger a_v  a_{-v}|\al\rangle\right|
\eeq
The right hand side can be divided into two cases: 
\beqa\label{rhs.decompose}
\sum_{\al\in N_v,\be\in N_u, \be(u)\be(-u)\geq \al(v)\al(-v)}\left|f(\al)f(\be)\langle \be|a_u^\dagger a_{-u}^\dagger a_v  a_{-v}|\al\rangle\right|\\\nonumber
+\sum_{\al\in N_v,\be\in N_u, \al(v)\al(-v)>\be(u)\be(-u) }\left|f(\al)f(\be)\langle \be|a_u^\dagger a_{-u}^\dagger a_v  a_{-v}|\al\rangle\right|
\eeqa
By definition of $f$, if $\langle \be|a_u^\dagger a_{-u}^\dagger a_v  a_{-v}|\al\rangle\neq 0$,  we have $|f(\be)|=|\lambda_u/\lambda_vf(\al)|$,  $\be(u)=\al(u)+1$ and $\be(-u)=\al(-u)+1$. Denote by  $N_{v,u}\subset N_v$ the set 
\beq
N_{v,u}\equiv\{\al\in N_v:(\al(u)+1)(\al(-u)+1)\leq \al(v)\al(-v)\}
\eeq
Hence   we can bound  \eqref{rhs.decompose} by
\beqa\label{temp8.41}
&&|\sum_{\al\in N_v,\be\in N_v}f(\al)f(\be)\langle \be|a_u^\dagger a_{-u}^\dagger a_v  a_{-v}|\al\rangle|\\\nonumber
\leq&&\sum_{\al \in N_{v,u}}\left | \frac{\lambda_u}{\lambda_v} \right | |f(\al)|^2\al(v)\al(-v)+\sum_{\be \in N_{u,v}}\left | \frac{\lambda_v}{\lambda_u} \right | |f(\be)|^2\be(u)\be(-u)
\eeqa

Now we bound $\sum_{\al \in N_{v,u}}|f(\al)|^2\al(v)\al(-v)$. If $\al   \in N_v$ and $\al(v)\al(-v)>0$, then with Proposition  \ref{prop2}, there exist $\al'$, $v'\in P_L$ with  $\al'\in M_{v'}^s$ such that 
\beq\label{alv'2}
\al=\mathcal A^{v', \,v-\frac{v'}{2}}\al'
\eeq
If $\al'\notin N_v$, then there exists $\gamma'$ s.t. $\mathcal A^{v}\gamma'=\al'$. Hence 
$$\mathcal A^v(\mathcal A^{v', \,v-\frac{v'}{2}}\gamma')=\al\Rightarrow\al\notin N_v
$$
and we have a contradiction. 
Hence  we have $\al'\in N_v$ and $\al'(-v)>0$. Again by  Proposition  \ref{prop2}, there exist $\al''$, $v''\in P_L$ such that  $\al''\in M_{v''}^s$ 
\beq\label{al'v''2}
\al'=\mathcal A^{v'', \,-v-\frac{v''}{2}}\al''
\eeq
Combining \eqref{alv'2} and \eqref{al'v''2} and using \eqref{properf4}, we  express $f(\al)$ in terms of  $f(\al'')$, $\al''(v')$, $\al''(v'')$ and $\al''(0)$ and $\lambda$'s. By definition of $M$,
 $\al''(\widetilde v)\leq m_c$  for any $\widetilde v\in P_L$ and  we obtain 
\beq\label{temp8.44}
|f(\al)|^2\leq \const \rho^2 m_c^2|\Lambda|^{-2}\lambda_v^2\left|\lambda_{-v+v'}\lambda_{v+v''}\right|f(\al'')^2
\eeq
By \eqref{ublambdaPH0} and $-v+v',v+v''\in P_H$, we have  
\beq
|f(\al)|^2\leq  \const \rho^2m_c^2\lambda_v^2 \eps_H^{-4}f(\al'')^2
\eeq 
Summing  over $v'$, $v''\in P_L$ and $\al''\in M$, we obtain  
\beq
\sum_{\al \in N_v, \al(v)+\al(-v)\geq2}|f(\al)|^2\leq \const \rho^5\eta_L^{-6} m_c^2\lambda_v^2 \eps_H^{-4}\leq (\rho^2\lambda_v)^2
\eeq 
Similarly, one can prove that  
\beq
\sum_{\al\in N_v, \al(v)+\al(-v)\geq m}|f(\al)|^2\leq  (\rho^2\lambda_v)^m
\eeq
Hence, we can obtain
\beq\nonumber
\sum_{\al\in N_{v,u} }\left|\frac{\lambda_u}{\lambda_v} \right||f(\al)|^2\al(v)\al(-v)\leq\sum_{\al\in N_v }\left|\frac{\lambda_u}{\lambda_v} \right||f(\al)|^2\al(v)\al(-v)\leq 2\rho^4\left|\lambda_u\lambda_v \right|
\eeq
Inserting this result into \eqref{temp8.41} and using the symmetry, we obtain
\beq 
|\sum_{\al\in N_v,\be\in N_v}f(\al)f(\be)\langle \be|a_u^\dagger a_{-u}^\dagger a_v  a_{-v}|\al\rangle|\leq \const \rho^4\left|\lambda_u\lambda_v\right|
\eeq
This completes the proof.
\end{proof}
Using this lemma, we can estimate the term $\langle a^\dagger_ua^\dagger_{-u}a_va_{-v}\rangle$ as follows.

\begin{lem}\label{u-uv-vPsi} 
For $u,v\in P_I\cup P_H$,
\beqa\label{u-uv-vPMH}
&&\left|\langle a_u^\dagger a_{-u}^\dagger a_va_{-v}\rangle-\lambda_u\lambda_v\frac{Q_\Psi(0,0)-Q_\Psi(0)}{|\Lambda|^2}\right|\\\nonumber
\leq &&\left|\lambda_u\lambda_v\right|\rho^2(\left(Q_\Psi(u,v)+Q_\Psi(u,-v)\right)/2+Q_\Psi(u)+Q_\Psi(v)+\const\rho^2)
\eeqa 
 For $u\in P_L$, $v\in P_I\cup P_H$,
\beqa\label{u-uv-vPLMH}
&&\left|\langle a_u^\dagger a_{-u}^\dagger a_va_{-v}\rangle-\lambda_u\lambda_v\left(\frac{Q_\Psi(0,0)-Q_\Psi(0)}{|\Lambda|^2}+\frac{\rho^4\lambda_u^2}{1-\rho^2\lambda_u^2}\right)\right|\\\nonumber
\leq &&\left|\lambda_u\lambda_v\right|\rho^2\big(\left(Q_\Psi(u,v)+Q_\Psi(u,-v)\right)/2+2Q_\Psi(v)\\\nonumber
+&&\frac{4\rho^2\lambda_u^2}{1-\rho^2\lambda_u^2}(\rho^{\eta/2}+(\rho\lambda_u)^{2\sqrt{m_c}})\big)
\eeqa 
 For $u,v\in P_L$
\beqa\label{u-uv-vPL}
&&\langle a_u^\dagger a_{-u}^\dagger a_va_{-v}\rangle-\lambda_u\lambda_v\frac{Q_\Psi(0,0)-Q_\Psi(0)}{|\Lambda|^2}\\\nonumber
\leq &&\left|\lambda_u\lambda_v\right|\rho^2\left(Q_\Psi(u,v)+2Q_\Psi(u)+2Q_\Psi(v)+3\right)
\eeqa
We note that there is no absolute value on the left hand side  of  the inequality  when $u,v\in P_L$.
\end{lem}

\begin{proof}
We first prove \eqref{u-uv-vPMH} concerning $u,v\in P_I\cup P_H$.  By  Lemma \ref{id4a2}, we have  
\beq\label{u-uv-vPMHproof1}
\left|\langle a_u^\dagger a_{-u}^\dagger a_va_{-v}\rangle-P(u,v)\right|\leq \const \rho^4\left|\lambda_u\lambda_v\right|,
\eeq
where $P(u,v)$ is defined in \eqref{defPuv}. By the property of $f$ in \eqref{properf1}, we can rewrite $P(u,v)$ as 
\beqa\label{u-uv-vPMHproof2}
\sum_{\gamma\in M, \mathcal A^{u}\gamma\in M, \mathcal A^{v}\gamma\in M}&&\lambda_u\lambda_v|f(\gamma)|^2\frac{\gamma(0)^2-\gamma(0)}{|\Lambda|^2}\\\nonumber
\times&&\sqrt{(\gamma(u)+1)(\gamma(-u)+1)(\gamma(v)+1)(\gamma(-v)+1)}
\eeqa
The situation that $\gamma\in M$ and $\mathcal A^{u(v)}\gamma\notin M$ can only happen when $\gamma(0)=1$ or $0$. But in this case,  $\gamma(0)^2-\gamma(0)=0$ and the term vanishes.  Hence the summation 
of $\gamma$ in \eqref{u-uv-vPMHproof2} can be replaced by $\sum_{\gamma\in M}$. Therefore, 
for  $u,v\in P_I\cup P_H$, we have 
\beqa\label{u-uv-vPMHproof3}
&&\left|P(u,v)-\lambda_u\lambda_v\frac{Q_\Psi(0,0)-Q_\Psi(0)}{|\Lambda|^2}\right|\\\nonumber
\leq&&\sum_{\gamma\in M}\left|\lambda_u\lambda_v\right||f(\gamma)|^2\frac{\gamma(0)^2}{|\Lambda|^2}\\\nonumber
&&\times\left|\sqrt{(\gamma(u)+1)(\gamma(-u)+1)(\gamma(v)+1)(\gamma(-v)+1)}-1\right|
\eeqa
{F}rom  $\gamma(0)\leq N$ and the Schwarz inequality,  the rhs. is bounded by  
\beq\label{u-uv-vPMHproof3.5}
\sum_{\gamma\in M}\left|\lambda_u\lambda_v\right||f(\gamma)|^2\rho^2
\left [ \left( \frac {\gamma(u)+\gamma(-u)}  2+1\right)\left(\frac {\gamma(v)+\gamma(-v)} 2 +1 \right) -1 \right ] 
\eeq
By symmetry, we have $Q_\Psi(u)=Q_\Psi(-u)$ and $Q_\Psi(u,v)=Q_\Psi(-u,-v)$. So we have
\beq
\eqref{u-uv-vPMHproof3.5}\leq\left|\lambda_u\lambda_v\right|\rho^2\left(\frac12\left(Q_\Psi(u,v)+Q_\Psi(u,-v)\right)+Q_\Psi(u)+Q_\Psi(v)\right)
\eeq
Together with \eqref{u-uv-vPMHproof1}, we have proved   \eqref{u-uv-vPMH}.

We now prove \eqref{u-uv-vPLMH} concerning $u\in P_L$, $v\in P_I\cup P_H$. Following arguments in the previous paragraph and using  \eqref{properf2} and \eqref{properf3}, we can rewrite $P(u,v)$ as
\beqa\label{u-uv-vPLMHproof1}
\sum_{\gamma\in M, \mathcal A^{u} \gamma \in M}&&\lambda_u\lambda_v|f(\gamma)|^2\frac{\gamma(0)^2-\gamma(0)}{|\Lambda|^2}\\\nonumber
&&\times\sqrt{(\gamma^*(u)+1)(\gamma^*(-u)+1)(\gamma(v)+1)(\gamma(-v)+1)}
\eeqa 
Notice that no matter we use \eqref{properf2} or \eqref{properf3}, the final result is the same. 
For $\gamma\in M$ with $\gamma(0)\geq 2$,  the case $\mathcal A^{u}\gamma\notin M$  can only happen when $\gamma^*(u)=m_c$. Hence, the summation of $\gamma$ in \eqref{u-uv-vPLMHproof1} can be replaced by $\sum_{\gamma^*(u)\neq m_c}$. Since  $\gamma^*(u)=\gamma^*(-u)$, for  $u\in P_L,v\in P_I\cup P_H$ we have 
\beq
P(u,v)=\!\!\!\!\!\sum_{\gamma^*(u)\neq m_c}\!\!\!\!\!\lambda_u\lambda_v|f(\gamma)|^2\frac{\gamma(0)^2\!-\!\gamma(0)}{|\Lambda|^2}(\gamma^*(u)+1)\sqrt{(\gamma(v)+1)(\gamma(-v)+1)}
\eeq
%
Since $\gamma(0)\leq N$, we have
\beqa\nonumber
&&\!\!\!\left|P(u,v)-\lambda_u\lambda_v\frac{Q_\Psi(0,0)\!-\!Q_\Psi(0)}{|\Lambda|^2}-\sum_{\gamma}\lambda_u\lambda_v|f(\gamma)|^2\frac{\gamma(0)^2-\gamma(0)}{|\Lambda|^2}\gamma^*(u)\right|
\\\nonumber
\leq&&\sum_{\gamma^*(u)\neq m_c}\left|\lambda_u\lambda_v\right||f(\gamma)|^2\rho^2(\gamma^*(u)+1)\left|\sqrt{(\gamma(v)+1)(\gamma(-v)+1)}-1\right|\\\label{u-uv-vPLMHproof2}
+&&\sum_{\gamma^*(u)=m_c}\left|\lambda_u\lambda_v\right||f(\gamma)|^2\rho^2\left(\gamma^*(u)+1\right)
\eeqa
We can replace $\sum_{\gamma^*(u)\neq m_c}$ in the first term of rhs.  by  $\sum_{\gamma\in M}$ to have an upper bound. Since  $\sqrt{(\gamma(v)+1)(\gamma(-v)+1)}-1\le [\gamma(v)+\gamma(-v)]/2$ and $\gamma^*(u)\leq \gamma(u)+1$, we can bound 
the right hand side of \eqref{u-uv-vPLMHproof2} by 
\beqa\label{u-uv-vPLMHproof3}
\left|\lambda_u\lambda_v\right|\rho^2\left[\half\left(Q_\Psi(u,v)+Q_\Psi(u,-v)\right)+2Q_\Psi(v)+\!\!\!\!\sum_{\gamma(u)\geq m_c-1}\!\!\!\!2|f(\gamma)|^2\gamma(u)\right]
\eeqa
The last term   is bounded  in \eqref{ubpsiuPLmc}, i.e., 
\beq\label{u-uv-vPLMHproof3.5}
\sum_{\gamma(u)\geq m_c-1}|f(\gamma)|^2\gamma(u)\leq \frac{\rho^2\lambda^2}{1-\rho^2\lambda^2}\, \rho^{\eta/2}
\eeq

The estimate \eqref{u-uv-vPLMH} follows from last three inequalities and \eqref{u-uv-vPLMHproof1},   provided 
that we can establish the following estimate
\beq\label{8.2}
\sum_{\gamma} |f(\gamma)|^2\frac{\gamma(0)^2-\gamma(0)}{|\Lambda|^2}\gamma^*(u) 
=\frac{\rho^{4}\lambda_u^2}{(1-(\rho\lambda_u)^2)}[1 + O(\rho^{\eta/2})+O((\rho\lambda_u)^{2\sqrt{m_c}})]
\eeq
To prove this, we first divide the summation of $\gamma$ into $\gamma\in M_u^s$ and $\gamma\in M_u^a$. 
For the case $\gamma\in M_u^s$, we have 
\beq
\sum_{\gamma \in M_u^s}|f(\gamma)|^2\frac{\gamma(0)^2-\gamma(0)}{|\Lambda|^2}\gamma^*(u)\leq\rho^2Q_\Psi(u)
\le  \rho^2\frac{(\rho\lambda_u)^2}{(1-(\rho\lambda_u)^2)}(1+\rho^{2/3}),
\eeq
where we have used \eqref{ubpsiuroughPL} in the last inequality. 
For the case $\gamma\in M_u^a$, using \eqref{ubpsiuroughPL1}, we have 
\beqa\nonumber
\sum_{\gamma \in M_u^a}|f(\gamma)|^2\frac{\gamma(0)^2-\gamma(0)}{|\Lambda|^2}\gamma^*(u)&&\leq
\const\rho^2 \frac{\rho m_{c}}{\eps_H}\sum_{\gamma\in M_u^s}|f(\gamma)|^2\gamma(u)\\
&&\le  \rho^{\frac{8}{3}}\frac{(\rho\lambda_u)^2}{(1-(\rho\lambda_u)^2)}.
\eeqa
This proves the upper bound part of \eqref{8.2}. The lower bound 
follows from  \eqref{lbpsiuPL0} since $\gamma^*(u) \ge \gamma(u)$.

Finally,  we prove  \eqref{u-uv-vPL} concerning $u,v\in P_L$. Similar to the previous argument, by \eqref{properf2} and \eqref{properf3}, we can rewrite $P(u,v)$ as
\beqa\label{u-uv-vPLproof1}
\sum_{\gamma\in M, \mathcal A^{u}  \gamma \in M, \mathcal A^{v}  \gamma \in M}  && \lambda_u\lambda_v|f(\gamma)|^2\frac{\gamma(0)^2-\gamma(0)}{|\Lambda|^2} \\
&& \times  \sqrt{(\gamma^*(u)+1)(\gamma^*(-u)+1)(\gamma^*(v)+1)(\gamma^*(-v)+1)}\nonumber 
\eeqa 
Since  $\lambda_u\lambda_v\geq0$ and $\gamma^*(u)=\gamma^*(-u)$, we have for $u,v\in P_L$,
\beq
P(u,v)-\lambda_u\lambda_v\frac{Q_\Psi(0,0)-Q_\Psi(0)}{|\Lambda|^2}\leq \sum_{\gamma\in M} \lambda_u\lambda_v\rho^2\left|(\gamma^*(u)+1)(\gamma^*(v)+1)-1\right|
\eeq
Using $\gamma^*-\gamma\leq 1$, we have proved  \eqref{u-uv-vPL}.
\end{proof}

We now can now prove Lemma \ref{lemHS3}. 
%

\begin{proof}
Summing over $u,v\neq 0$ of \eqref{u-uv-vPMH}, \eqref{u-uv-vPLMH} and \eqref{u-uv-vPL}, we obtain that 
\beq
\sum_{u,\,v\neq 0}
\frac{V_{u-v}}{|\Lambda|^2}\langle a^\dagger_ua^\dagger_{-u}a_va_{-v}\rangle
\le A+ B + \Omega
\eeq
where
\[
A= \frac{Q_\Psi(0,0)-Q_\Psi(0)}{|\Lambda|^2}  \sum_{u,v\neq 0}\frac{V_{u-v}}{|\Lambda|^2}\lambda_u\lambda_v
\]
\[
B= 2\sum_{u\in P_L,v\in P_I\cup P_H}\frac{V_{u-v}}{|\Lambda|^2}\lambda_u\lambda_v\frac{\rho^4\lambda_u^2}{1-\rho^2\lambda^2}
\]
\beqa\label{longexpression1}
\Omega &= &\frac{1}{|\Lambda|^2}\bigg(\sum_{u,\,v\neq 0}|V_{u-v}||\lambda_u\lambda_v|\rho^2Q_\Psi(u,v)+\!\!\!\sum_{u\in P_I\cup P_H,v\neq 0}\!\!\!\!\!\!4|\lambda_u\lambda_vV_{u-v}|\rho^2Q_\Psi(u)\nonumber \\
&+&\sum_{u,\,v\in P_L}3|V_{u-v}||\lambda_u\lambda_v|\rho^2(Q_\Psi(u)+1)+\sum_{u,v\in P_I\cup P_H}\const\rho^4|\lambda_u\lambda_v||V_{u-v}|\nonumber \\
&+&\sum_{u\in P_L,v\in P_I\cup P_H}|\lambda_u\lambda_v||V_{u-v}|\frac{4\rho^4\lambda_u^2}{1-\rho^2\lambda_u^2}(\rho^{\eta/2}+(\lambda_u\rho)^{2\sqrt{m_c}})\bigg)
\eeqa

The error term $\Omega$ can be bounded by using  the following facts, (1): $|\rho\lambda_u|\leq 1$, (2): $|\sum_{v\neq 0}\lambda_vV_{u-v}|\leq \const\Lambda$, (3): $|V_u|\leq V_0$, (4): $|\lambda_u|\leq g_0|u|^{-2}$for any $u\neq 0$ and (5): $\sum_{u,v}|\lambda_uV_{u-v}\lambda_v|\leq \const |\Lambda|^2$:
\beqa\nonumber
\Omega  &\le & \frac{\const}{|\Lambda|^2}\bigg(\sum_{u,\,v\neq 0}Q_\Psi(u,v)+\sum_{u\in P_I\cup P_H} Q_\Psi(u)\rho\Lambda+\sum_{u,\,v\in P_L}\frac{Q_\Psi(u)+1}{u^2v^2}\rho^2\\\label{temp8.70}
&+&\rho^4|\Lambda|^2+\sum_{u\in P_L}\frac{\rho^3\lambda_u^2}{1-\rho^2\lambda_u^2} \Lambda (\rho^{\eta/2}+(\lambda_u\rho)^{2\sqrt{m_c}})\bigg)
\eeqa
By  \eqref{boundspsiuvtotal} and \eqref{boundspsiuPMPHtotal}, the first two terms on the right hand side 
are bounded by $o(\rho^{5/2})$. Using the trivial bound  $Q_\Psi(u)\leq m_c$ for $u\in P_L$, 
the third term is also bounded by $o(\rho^{5/2})$. By  \eqref{temp8.27} and \eqref{temp8.28}, the last term is also $o(\rho^{5/2})$.  Hence the error terms are bounded by $\Omega \leq o(\rho^{5/2})$.

We now estimate $A$ and $B$.  Notice that
$
(Q_\Psi(0,0)-Q_\Psi(0))|\Lambda|^{-2}
= \rho_{0}^{2} + o(\rho^{5/2})
$. 
Hence we shall replace this factor in $A$ by $\rho_{0}^{2}$.  
Since  $\lambda_{u}= - w_{u}$ for $u \in P_{I}\cup P_{H}$, we have 
$$\sum_{u,v\neq 0}\lambda_u\lambda_v=\sum_{u,v\neq 0}w_uw_v-2\sum_{u\in P_L,v\not = 0}(\lambda_u+w_u)w_v+\sum_{u, v\in P_L}(\lambda_u+w_u)(\lambda_v+w_v)
$$
We can now decompose $A$ into 
\beq \label{temp8.73}
A = \|w^2V\|_1 \rho_{0}^{2}+A_{1}+A_{2}+A_{3} + o(\rho^{5/2}) 
\eeq
where 
\[
A_{1}=
-2\rho_{0}^{2} \sum_{u\in P_L,v\in P_I\cup P_H}\frac{V_{u-v}}{|\Lambda|^2}(\lambda_u+w_u)w_v
\]
\[
A_{2}=
-2\rho_{0}^{2} \sum_{u\in P_L,v\in P_{L}}\frac{V_{u-v}}{|\Lambda|^2}(\lambda_u+w_u)w_v
\]
\[
A_{3}= \sum_{u,v\in P_L}\frac{V_{u-v}}{|\Lambda|^2}(\lambda_u+w_u)(\lambda_v+w_v)\rho_{0}^{2}
\]
Since  $|w_u\rho|\leq \const\rho |u|^{-2}\leq \eps_L^{-2}$,  we have $A_{3}\le o(\rho^{5/2})$.
We can also obtain the simple estimate $A_{2}\le o(\rho^{5/2})$.

If we replace $\rho^{2}$ in $B$ by  $\rho_{0}^{2}$, which is equal to $\rho^2-O(\rho^{5/2})$, we have 
\beq
B+ A_{1}=
-2\sum_{u\in P_L,v\in P_I\cup P_H}\frac{V_{u-v}}{|\Lambda|^2}w_v\rho^2\left(\lambda_u\frac{\rho^2\lambda_u^2}{1-\rho^2\lambda_{u}^2}+\lambda_u+w_u\right) 
\eeq
Using  $|V_{u-v}-V_v|\leq \const |u|$ for $u \in P_{L}$ and $v \in P_{I}\cup P_{H}$,  we can simplify $B+A_1$ as
\beq\label{temp8.75}
B+A_1\leq -\frac{2\|Vw\|_1}{|\Lambda|}\rho^2\sum_{u\in P_L}\left(\frac{\lambda_u}{1-\rho^2\lambda_{u}^2}+w_u\right)+o(\rho^{5/2})
\eeq

Since  $|g_u-g_0|\leq \const|u|$, we have $| w_u- g_0|u|^{-2}|\leq \const \rho^{-1/2}\eps_L^{-1}$. Then we can replace $w_u$ with $g_0|u|^{-2}$ in \eqref{temp8.75}. Setting  $u=\rho^{1/2}k$, we have, by definition of $\lambda$, 
\beqa
&&\lim_{\rho\to 0}(\rho^{1/2}\Lambda)^{-1}\sum_{u\in P_L}\left(\frac{\lambda_u}{1-\rho^2\lambda^2_u}+g_0|u|^{-2}\right)\\\nonumber
=&&\frac{1}{8\pi^3}\int_{k\in \R^3}g_0|k|^{-2}\left(\frac{\sqrt{1+4g_0|k|^{-2}}-1}{\sqrt{1+4g_0|k|^{-2}}}\right)dk^3=\frac{g_0^{3/2}}{\pi^2}
\eeqa
Inserting this result  into \eqref{temp8.75} and \eqref{temp8.73}, we have proved \eqref{proofthem1-4}.
\end{proof}
\section{Proof of Lemma \ref{lemHAS1}}

In this section, we  prove Lemma \ref{lemHAS1} concerning  potential energy terms with one $a_{0}$. 
Let  $v_j\in \Lambda^*$ and $v_{j}\neq 0$ for $j = 1, 2, 3$.  Define $P_{H,\,c}$ as the following  subset of $P_H$:
\beq\label{defPHc}
P_{H,\,c}=\{k\in P_H: |k|\leq k_c\}.
\eeq
The following lemma classify all possible scenarios of $v_1,\,v_2,\,v_3$. Through out  this section, we assume that $v_{i} \not = 0$
for $i=1, 2, 3$. 

\begin{lem}\label{posscombone0}
Suppose  $\be, \alpha \in M$ and $\langle \al| a^\dagger_{0}a^\dagger_{v_1}a_{v_2}a_{v_3}|\be\rangle\neq 0$. 
Then there are only three possibilities:
\begin{enumerate}
	\item 
	\beq\label{one0firstcase}
	v_1\in P_L,\,\,\,v_2, v_3\in P_{H,\,c},   v_{i} \not = \pm v_{j} \text{ for } i \not = j .
	\eeq
		\item 
	\beq\label{one0secondcase}
	v_1\in P_{H,\,c}, \,\,\,v_{2}\in P_L, \,\,\,v_{3}\in P_{H,\,c},  v_{i} \not = \pm v_{j} \text{ for } i \not = j; \text{ \rm or }  2 \leftrightarrow 3.
	\eeq
	\item  
	\beq\label{one0thirdcase}
	v_1\in P_L,\,\,\,v_2\in P_L,\,\,\,v_3\in P_L.
	\eeq
\end{enumerate}
\end{lem}
\begin{proof}
Since particles with momenta in   $P_I$ are always created in pair, e.g.,  $(u,-u)$, 
either none of $v_i$'s belongs to $P_I$ or two of them belong to $P_I$. Thus  we have:
 \beq 
 v_1,v_{2}\in P_I\Rightarrow v_1=v_{2}, \text{ \rm or }  2 \leftrightarrow 3
 \eeq
 \beq
 v_2,v_3\in P_I\Rightarrow v_2=-v_3.
 \eeq
If two of $v_i$'s are in $P_I$, by the momentum conservation $v_1=v_2+v_3$ the other one must be equal to zero, which is  a  contradiction. Therefore
\beq
 v_i\notin P_I,\,\,\,{\rm for}\,\,\,1\leq i\leq 3 \, .
\eeq  
The restriction  $|v_i|\leq k_c$ follows from the construction of $M$. 
 Therefore, we have
\beq\label{PLPHc}
 v_i\in P_L\cup P_{H,c},\,\,\,{\rm for}\,\,\,1\leq i\leq 3
 \eeq  
 
Since particles in  $P_{H,c}$ are always created in soft pair creations which generated two particles in  
$P_{H,c}$, the number of particles in $P_{H,c}$ is even. So either none of $v_i$'s are in $P_{H,c}$ or two of them are in $P_{H,c}$. Together with \eqref{PLPHc}, and momentum conservation, we prove the lemma.
 \end{proof}

For fixed $v_1,v_2,v_3$, define 
\beq\label{defFal}
F(\al)\equiv \sum_{i: v_i\in P_L, i = 1, 2, 3} |\al(v_i)-\al(-v_i)|
\eeq
\begin{lem}\label{FalFbe0}
For any $\al, \be\in M$  if  $ \langle \al| a^\dagger_{0}a^\dagger_{v_1}a_{v_2}a_{v_3}|\be\rangle\neq 0$ and $v_i\neq \pm v_j$, we have:
 \beq\label{fab=no}
 F(\al)+F(\be)= \# \{i =1, 2, 3:  v_i\in P_L \}
 \eeq  
Furthermore,  the ratio between  $f(\al)$ and $f(\be)$ is bounded as follows. 
 \beq\label{relationalbe}
 \rho^{\frac{1}{20}}\sqrt{N}^{F(\al)-F(\be)}\leq \left|\frac{f(\be)\sqrt{\lambda_{v_1}}}{f(\al)\sqrt{\lambda_{v_2}\lambda_{v_3}\al(0)/\Lambda}}\right|\leq \rho^{\frac{-1}{20}}\sqrt{N}^{F(\al)-F(\be)}
 \eeq
\end{lem}
 \begin{proof} Since  $v_i\neq \pm v_j$,  for each $i$ fixed,  if $\alpha \in M^a_{v_i}$, 
 then $\beta \in  M^s_{v_i}$ and vice verse.    This proves  \eqref{fab=no}.

 Recall  the definition of $f$ in \eqref{deffal}. Then one can check the ratio involving $f(\be)/f(\al)$ in \eqref{relationalbe}  depends only on the last factor 
 \[
 \prod_{u\in P_L, \alpha^\ast (u) - \alpha (u) = 1}\sqrt{4\al^*(u)\lambda_u|\Lambda|^{-1}}
\]
We now use \eqref{lublambdaPL0} 
to bound $\lambda$ in this expression. Since $F(\alpha)$ counts how many times this factor appears, 
this proves \eqref{relationalbe}. 
 \end{proof}

Using the definitions of $\eta_L$ and $m_c$, the bound $\alpha(0)/\Lambda \le \rho$ 
and lemma \ref{id4a}, we have   
 \beqa\label{fal0}
 &&\left|f(\al)f(\be)\langle \al| a^\dagger_{0}a^\dagger_{v_1}a_{v_2}a_{v_3}|\be\rangle\right|\\\nonumber
 \leq&& \sqrt{N}^{F(\al)-F(\be)+1}\rho^{\frac{-1}{20}} \sqrt{ \rho \left|\frac{\lambda_{v_2}\lambda_{v_3}}{\lambda_{v_1}}\right|}\sqrt{\al(v_1)(\al(v_2)+1)(\al(v_3)+1)}|f(\al)|^2
 \eeqa
 and
 \beqa
 &&\label{fbe0}
 \left|f(\al)f(\be)\langle \al| a^\dagger_{0}a^\dagger_{v_1}a_{v_2}a_{v_3}|\be\rangle\right|\\\nonumber
 \leq&&\sqrt{N}^{F(\be)-F(\al)+1}\rho^{\frac{-1}{20}} \sqrt{\left|\lambda_{v_1}\lambda_{v_2}^{-1}\lambda_{v_3}^{-1}\right|\rho^{-1}}\sqrt{(\be(v_1)+1)\be(v_2)\be(v_3)}|f(\be)|^2
 \eeqa

Lemma \ref{lemHAS1} follows from summing the three inequalities of the following Lemma. 
 \begin{lem}
In the limit $k_c\to \infty$, $\rho\to 0$, we have 
 \beq\label{0resultfirstcase}
	\overline\lim_{k_c,\,\,\rho}  |\Lambda|^{-2} \rho^{-5/2}  \sum_{\eqref{one0firstcase}}\left\langle  V_{v_2}a^\dagger_0a^\dagger_{v_1}a_{v_2}a_{v_3}\right\rangle
	= - 2 \|Vw\|_1\frac{g_0^{3/2}}{3\pi^2}
	\eeq
	\beq\label{0resultsecondcase}
	\overline\lim_{k_c,\,\,\rho} |\Lambda|^{-2} \rho^{-5/2} \sum_{\eqref{one0secondcase}}\left|  V_{v_2} \left\langle a^\dagger_0a^\dagger_{v_1}a_{v_2}a_{v_3}\right\rangle\right|=0
	\eeq
	\beq\label{0resultthirdcase}
	\overline\lim_{k_c,\,\,\rho}|\Lambda|^{-2}\rho^{-5/2} \sum_{\eqref{one0thirdcase}}\left|  V_{v_2} \left  \langle a^\dagger_0a^\dagger_{v_1}a_{v_2}a_{v_3}\right\rangle\right|=0
\eeq
 \end{lem}
 \begin{proof} 
 We first  prove (\ref{0resultfirstcase}) concerning  \eqref{one0firstcase}, which implies that 
 $F(\alpha)+ F(\beta)=1$.   By the bounds on $\lambda_u$ in \eqref{lublambdaPL0} and $\al^*(u)\leq m_c$ for $u\in P_L$, we have, for $ F(\be)=0$  the following slightly modified version  of  \eqref{fbe0} 
 \beq\label{fbe0dup1}
 |f(\al)f(\be)|\left|\langle \al| a^\dagger_{0}a^\dagger_{v_1}a_{v_2}a_{v_3}|\be\rangle\right|\leq\rho^{\frac{-1}{10}} \rho^{-1}\sqrt{\left|\lambda_{v_2}^{-1}\lambda_{v_3}^{-1}\right|}\sqrt{\be(v_2)\be(v_3)}|f(\be)|^2
 \eeq
Here we replaced  $\rho^{-1/20}$ in \eqref{fbe0} by $\rho^{-1/10}$ to accommodate small errors. 
Summing over $\beta$ with $F(\be)=0$, we have 
\beq\label{sumfbe01}
\sum_{F(\be)=0}f(\be)f(\al)\left|\langle \al| a^\dagger_{0}a^\dagger_{v_1}a_{v_2}a_{v_3}|\be\rangle\right|\leq \rho^{-11/10}\sqrt{\left|\lambda_{v_2}^{-1}\lambda_{v_3}^{-1}\right|}Q_\Psi(u,v)
\eeq
Using the bound \eqref{boundpsifixeduv} on $Q_\Psi(u,v)$  and $|\lambda_u|\leq g_0|u|^{-2}$, we obtain that $\eqref{sumfbe01}=o(\rho^2)$. 

Since $F(\alpha)+ F(\beta)=1$, the other case is $F(\alpha)=0$.  
Hence we have 
\beqa\label{temp8.96}
&&\langle  a^\dagger_{0}a^\dagger_{v_1}a_{v_2}a_{v_3}\rangle= A_{1}+ A_{2}+ o(\rho^{3/2})\\\nonumber
&&A_{1} = \sum_{F(\al)=0}\sqrt{\al(0)\al(v_1)}f(\al)f(\be)\\\nonumber
&&A_{2} = \sum_{F(\al)=0}\sqrt{\al(0)\al(v_1)}\left(\sqrt{(\al(v_2)+1)(\al(v_3)+1)}-1\right)f(\al)f(\be)
\eeqa
By the estimate  \eqref{relationalbe} and the Schwarz inequality 
$ | 2 (\sqrt{(a+1)(b+1)}-1)|\leq a+b$, we  have 
\begin{align}\label{temp8.100}
|A_{2}| & \le \rho^{4/5}\sum_{F(\al=0)} \frac {\al(v_2)+\al(v_3)}  2 |f(\al)|^2 \\
& \leq  \rho^{4/5}(Q_\Psi(v_2)+Q_\Psi(v_3)) \leq o(\rho^2)\, , 
\end{align}
where we have used the bounds on $\lambda$'s and $Q_\Psi(u)$ for $u\in P_H$. 

By the property  \eqref{properf4} for $f$, we have 
\beq
A_{1}= 2 \sqrt{\lambda_{v_2}\lambda_{v_3}}\sum_{F(\al)=0}\al(0)\al(v_1)|\Lambda|^{-1}|f(\al)|^2
\eeq
We notice 
\beq
\sum_{F(\al)=0}\al(0)\al(v_1)|f(\al)|^2=Q_{\Psi}(0,v_1)|\Lambda|^{-1}-\sum_{\al\in M_{v_1}^a}\al(0)\al(v_1)|\Lambda|^{-1}|f(\al)|^2
\eeq

The absolute value of the second term is less than $\rho m_c \sum_{\al\in M_{v_1}^a}|f(\al)|^2$. By \eqref{boundMa}, it is less than $\rho^{7/4}$. Then  with $|\sqrt{\lambda_{v_2}\lambda_{v_3}}|\leq O(\eps_H^{-2})$, we obtain 
\beq\label{temp8.103}
\langle  a^\dagger_{0}a^\dagger_{v_1}a_{v_2}a_{v_3}\rangle= 2 \sqrt{\lambda_{v_2}\lambda_{v_3}}Q_{\Psi}(0,v_1)|\Lambda|^{-1}+o(\rho^{3/2}).
\eeq

Recall $\lambda_{u} = - w_{u}$ for $u \in P_{I}\cup P_{H}$ and $w_{u}= w_{-u}$ due to our assumption on $V$. 
Since  $v_{1} \le P_{L} \sim \sqrt \rho $ and $v_{2}= -v_{3} + v_{1}$ and $v_{2} \in P_{H, c}$, we can check 
that 
\beq
\left|\lambda_{v_2}-\lambda_{v_3}\right|\leq \rho^{1/3}
\eeq
Inserting this in \eqref{temp8.103}, we arrive at 
\beq
\left\langle a^\dagger_0a^\dagger_{v_1}a_{v_2}a_{v_3}\right\rangle= 2 \lambda_{v_2}Q_\Psi(0,v_1)|\Lambda|^{-1}+o(\rho^{5/4})
\eeq
In the limit $k_c\to\infty,\rho\to 0$, we have
\beq
|\Lambda|^{-2}\sum_{v_1\in P_L, v_2\in P_{H,c}}\left\langle V_{v_2} a^\dagger_0a^\dagger_{v_1}a_{v_2}a_{v_3}\right\rangle=-\|Vw\|_1 |\Lambda|^{-2} \sum_{v_1\in P_L}Q_\Psi(0,v_1)+o(\rho^{5/2})
\eeq
We note
\beq
|\Lambda|^{-2} \sum_{v_1\in P_L}Q_\Psi(0,v_1)=\rho |\Lambda|^{-1}Q_\Psi(0)- |\Lambda|^{-2}Q_\Psi(0,0)-
|\Lambda|^{-2} \sum_{u\in P_I\cup P_H}Q_\Psi(0,u)
\eeq
The last term is less than $N|\Lambda|^{-2} \sum_{u\in P_I\cup P_H}Q_\Psi(u)\le o(\rho^{5/2})$ by Theorem \ref{thmA}.  Together with Lemma \ref{boundspsi02}, \ref{boundspsi0} on $Q_\Psi(0,0)$ and $Q_\Psi(0)$, we can compute the first two terms, i.e.,
\beq\label{sum0PL}
|\Lambda|^{-2} \sum_{v_1\in P_L}Q_\Psi(0,v_1)=\rho_0(\rho-\rho_0)+o(\rho^{5/2})
\eeq
This yields  \eqref{0resultfirstcase}.

We next prove \eqref{0resultsecondcase} concerning  \eqref{one0secondcase}. Without loss of generality we assume  that 
\beq
v_{1,3}\in  P_{H,c}\rmand v_2\in P_L
\eeq
Following similar arguments  in the previous proof, i.e.,  using  Lemma \ref{id4a}, \eqref{fal0} or  \eqref{fbe0} and the bounds on $\lambda_u$'s, we have
\beqa
|\langle a^\dagger_0a^\dagger_{v_1}a_{v_2}a_{v_3}\rangle|&&\leq \sum_{F(\al)=0}\rho^{-\frac{1}{10}} \sqrt{\al(v_1)(\al(v_3)+1)\left|\lambda_{v_1}^{-1}\right|}|f(\al)|^2\\\nonumber
+&&\sum_{F(\be)=0}\rho^{-\frac{1}{10}} \sqrt{\be(v_3)(\be(v_1)+1)\left|\lambda_{v_3}^{-1}\right|}|f(\be)|^2
\eeqa
For the upper bound, we can replace $\sum_{F(\al)=0}$ by  $\sum_{\al \in M}$. Using the upper bounds (\ref{ubpsiuroughPH}) and  (\ref{boundpsifixeduv}) on $Q_\Psi(u)$ and  $Q_\Psi(u,v)$  for $u,v\in P_H$, we obtain 
$
\left|\left\langle a^\dagger_0a^\dagger_{v_1}a_{v_2}a_{v_3}\right\rangle\right| \leq \const\rho^{3/2}
$. 
This proves  (\ref{0resultsecondcase}).

We now prove \eqref{0resultthirdcase} concerning   $v_i\in P_L$ satisfying   $F(\al)+F(\be)=3$.   
It is easy to prove that the contribution from the special cases,  $v_1=-v_2$(or $v_3$) or $v_2=v_3$,  is negligible, 
\beq
\overline\lim_{\rho}\sum_{\rm special\,\,\,cases}
  \left |  V_{v_2} \langle a^\dagger_0a^\dagger_{v_1}a_{v_2}a_{v_3}\rangle \right | 
\rho^{-5/2}|\Lambda|^{-2}=0
\eeq
So from now on we assume that $v_i\neq \pm v_j$ for $i\neq j$.  As before, we rewrite  $\langle a^\dagger_0a^\dagger_{v_1}a_{v_2}a_{v_3}\rangle$ by using  Lemma \ref{id4a} and \eqref{fal0} or \eqref{fbe0}.  Together with the bounds on $\lambda_u$'s  and $\al(v_i)\leq m_c$, we have
\beqa
|\langle a^\dagger_0a^\dagger_{v_1}a_{v_2}a_{v_3}\rangle|
\leq &&\sum_{F(\al)=0}N^{-1}\rho^{-\frac{1}{10}} |f(\al)|^2+\sum_{F(\al)=1}\rho^{-\frac{1}{10}} |f(\al)|^2\\\nonumber
+&&\sum_{F(\be)=0}N^{-1}\rho^{-\frac{1}{10}} |f(\be)|^2+\sum_{F(\be)=1}\rho^{-\frac{1}{10}} |f(\be)|^2
\eeqa
By symmetry, we only need to estimate the first two terms on the rhs. The first term is less than $N^{-1}\rho^{-\frac{1}{10}}$. For the second term,  we note $F(\al)=1$ implies  that there exists $i, 1\leq i\leq 3$
such that  $\al \in M_{v_i}^a$.  By \eqref{boundMa}, we have
\beq
\sum_{F(\al)=1}|f(\al)|^2\leq \rho^{3/4}
\eeq
This implies $|\langle a^\dagger_0a^\dagger_{v_1}a_{v_2}a_{v_3}\rangle|\le \rho^{1/2}$ and \eqref{0resultthirdcase}, which 
complete the proof.
\end{proof}

\section{Interaction Energy with Four Nonzero Momenta: The Classification} 

In the next three sections, we will prove Lemma \ref{lemHAS2} involving interaction energy without $a_0$. 
We will show that the only contribution to the accuracy we need comes from four high momentum particles, 
to  be computed in next section. In this section, we start the procedure of identifying the error terms. 

For $\al, \be \in M$, we have the following lemma, similar to  Lemma \ref{posscombone0} and Lemma \ref{FalFbe0}. 
Since it can be proved by same method, we will only state the result. 

\begin{lem}\label{fourlegslemma1}
Suppose   $v_i\neq 0, 1\leq i\leq 4$ and  $v_{1} + v_{2} \not = 0$, $v_1\neq v_3$ or $v_4$. If  $\langle \al| a^\dagger_{v_{1}}a^\dagger_{v_2}a_{v_3}a_{v_4}|\be\rangle\neq 0$ for some $\alpha, \beta \in M$, then there are exactly four cases: 
\begin{enumerate}
	\item All of $v_i\in P_L$ for $1\leq i\leq 4$. 
	\item 
$	v_{1}, v_{2}\in P_{L}, \,\,\,v_{3}, v_{4}\in P_{H,\,c}$.
	\item One of $v_{1}, v_{2}$ is in $P_L$ and the other is in $P_{H,c}$;  one of $v_{3}, v_{4}$ is in $P_L$ and the other is in $P_{H,c}$.
	\item All of $v_i\in P_{H,c}$ for $1\leq i\leq 4$. 
\end{enumerate}
If $v_i\neq \pm v_j$, for $1\leq i,j\leq 4$, we have
\beq 
\rho^{\frac{1}{20}}\sqrt{N}^{F(\al)-F(\be)}\leq \left|\frac{f(\be)\sqrt{\lambda_{v_1}\lambda_{v_2}}}{f(\al)\sqrt{\lambda_{v_3}\lambda_{v_4}}}\right|\leq \rho^{\frac{-1}{20}}\sqrt{N}^{F(\al)-F(\be)},
\eeq
 \beqa\label{no0Fal}
&&\left|f(\al)f(\be)\langle \al| a^\dagger_{v_1}a^\dagger_{v_2}a_{v_3}a_{v_4}|\be\rangle\right|\\\nonumber
 &&\leq \sqrt{N}^{F(\al)-F(\be)}\rho^{\frac{-1}{20}} \sqrt{\frac{\lambda_{v_3}\lambda_{v_4}}{\lambda_{v_1}\lambda_{v_2}}}\sqrt{\al(v_1)\al(v_2)(\al(v_3)+1)(\al(v_4)+1)}|f(\al)|^2
\eeqa
and
\beqa\label{no0Fbe}
&&\left|f(\al)f(\be)\langle \al| a^\dagger_{v_1}a^\dagger_{v_2}a_{v_3}a_{v_4}|\be\rangle\right|\\\nonumber
 &&\leq \sqrt{N}^{F(\be)-F(\al)}\rho^{\frac{-1}{20}} \sqrt{\frac{\lambda_{v_1}\lambda_{v_2}}{\lambda_{v_3}\lambda_{v_4}}}\sqrt{(\be(v_1)+1)(\be(v_2)+1)\be(v_3)\be(v_4)}|f(\be)|^2.
\eeqa
\end{lem}

\begin{prop} For $u\in P_L$ and $v\in P_{H,c}$, we have the following inequality
\beq\label{temp8.119}
\sum_{\al\in M_u^a}\al(v)\left|f(\al)^2\right|\leq \left|\lambda_v\right|\rho^{3-\frac{1}{10}}
\eeq\end{prop}
\begin{proof} By definition of $M$ \eqref{defMsau}, for any $\al\in M_u^a$, there exist $\be\in M_u^s$ and $k$ such that $\mathcal A^{u,k}\be=\al$ and $\pm k+u/2\in P_{H,c}$. Clearly, for any $v \in P_{H}$ we have 
$\alpha(v) \le \beta(v) +1$ and the case we need the constant $1$ occurs only when  $v=k+u/2 $ or 
$v= - k+u/2$. Hence  we can bound the left hand side of \eqref{temp8.119} by
\beq\label{10.1}
 \sum_{\be}\sum_{k:\pm k+u/2\in P_{H,c}}\be(v)|f(\mathcal A^{u,k}\be)^2|+\sum_{\be}\sum_{k:\pm k+u/2=v}|f(\mathcal A^{u,k}\be)^2|
\eeq
Recall   \eqref{properf4} implies that 
\beq
|f(\mathcal A^{u,k}\be)|^2\leq |f(\be)|^2\rho m_c|\Lambda|^{-1}\left|\lambda_{k+u/2}\lambda_{-k+u/2}\right|
\eeq
By  the bound \eqref{ublambdaPH0} on $\lambda$, we obtain that
\beqa\nonumber
\eqref{10.1}&&\leq \sum_{\be}\be(v)|f(\be)|^2\frac{\rho m_c}{|\Lambda|}\left[\sum_{\pm k+u/2\in P_H}\left|\lambda_{k+u/2}\lambda_{-k+u/2}\right|\right]+\left|\lambda_v\right||\Lambda|^{-1}\\
&&\leq Q_\Psi(v)\rho m_c\eps_H^{-4}k_c^3+\left|\lambda_v\right||\Lambda|^{-1}
\eeqa
Using Proposition \ref{particlebound}, we have proved  \eqref{temp8.119}.
\end{proof} 

\begin{lem}We have the following estimates on the interaction energies: 
\beq\label{resultasymno24L}
\overline\lim_{m_c,\,\rho}
 \; \rho^{-5/2} {|\Lambda|}^{-2}  \sum_{v_1,v_2,v_3,v_4\in P_L} \left | V_{v_{1}-v_{3}}\left\langle a_{v_1}^\dagger a_{v_2}^\dagger a_{v_3}a_{v_4}\right\rangle_{\Psi}\right| =0,
\eeq
\beq\label{resultasymno22L2H1}
\overline\lim_{m_c,\,\rho}
\; \rho^{-5/2} {|\Lambda|}^{-2}  \sum_{v_1+v_2\neq0, v_1,v_2\in P_L,v_3,v_4\in P_H}  \left | V_{v_{1}-v_{3}}\left\langle a_{v_1}^\dagger a_{v_2}^\dagger a_{v_3}a_{v_4}\right\rangle_{\Psi}\right|  = 0
\eeq
\beq\label{resultasymno22L2H2}
\overline\lim_{m_c,\,\rho}
 \; \rho^{-5/2} {|\Lambda|}^{-2}  \sum_{v_1,v_3\in P_L,v_2,v_4\in P_H} \left | V_{v_{1}-v_{3}}\left\langle a_{v_1}^\dagger a_{v_2}^\dagger a_{v_3}a_{v_4}\right\rangle_{\Psi}\right|  =0
\eeq
In other words, the contributions from case 1, 2 and 3 in Lemma \ref{fourlegslemma1} are negligible for our purpose. 
\end{lem}
\begin{proof}
We first prove the \eqref{resultasymno24L} concerning  $v_i\in P_L$. By Lemma \ref{id4a}, we have 
\beq
\left|\left\langle a_{v_1}^\dagger a_{v_2}^\dagger a_{v_3}a_{v_4}\right\rangle_{\Psi}\right|\leq \sum_{\al}\left|f(\al)
f(T(\al))\right| m_c^4
\eeq
Using the Schwarz inequality, we have $\left|\left\langle a_{v_1}^\dagger a_{v_2}^\dagger a_{v_3}a_{v_4}\right\rangle_{\Psi}\right|\leq m_c^4$. The summation  over the $v_i$ with $v_i=\pm v_j$ for some $1\leq i<j\leq 4 $
is negligible in the sense that  
\beq\label{specialviPL}
{|\Lambda|}^{-2} \sum_{v_1,v_2,v_3,v_4\in P_L, v_i=\pm v_j} \left | \left\langle a_{v_1}^\dagger a_{v_2}^\dagger a_{v_3}a_{v_4}\right\rangle_{\Psi}  \right | \leq o(\rho^{5/2})
\eeq
{F}rom now on, we assume that  $v_i\neq \pm v_j$ for any  $1\leq i<j\leq 4 $.  

Using \eqref{no0Fal}, \eqref{no0Fbe}  and the bounds  \eqref{lublambdaPL0} on $\lambda$, we have 
\beqa\nonumber
\left|\left\langle a_{v_1}^\dagger a_{v_2}^\dagger a_{v_3}a_{v_4}\right\rangle_{\Psi}\right|&&\leq \sum_{F(\al)\leq 1}\rho^{\frac{-1}{10}}N^{-1}|f(\al)^2|\\\nonumber
&&+\sum_{F(\al)=2}\rho^{\frac{-1}{10}}|f(\al)^2|+\sum_{F(\be)\leq 1}\rho^{\frac{-1}{10}}N^{-1}|f(\be)^2|
\eeqa
By \eqref{boundMaMa}, we have $\left|\left\langle a_{v_1}^\dagger a_{v_2}^\dagger a_{v_3}a_{v_4}\right\rangle_{\Psi}\right|\leq \rho^{9/5} .  $ 
Together with \eqref{specialviPL} and  $\Lambda=\rho^{-25/8}$, we can sum over $v_{j}$ to have  
\beq
{|\Lambda|}^{-2} \sum_{v_1,v_2,v_3,v_4\in P_L}\left\langle a_{v_1}^\dagger a_{v_2}^\dagger a_{v_3}a_{v_4}\right\rangle_{\Psi}\leq o(\rho^{5/2})
\eeq

We now prove \eqref{resultasymno22L2H1} concerning  $v_{1,2}\in P_L$ and $v_{3,4}\in P_{H_c}$.  As before, by \eqref{no0Fal}, \eqref{no0Fbe},  \eqref{lublambdaPL0} and \eqref{ublambdaPH0}, we have 
\beqa\nonumber
&&\left|\left\langle a_{u_1}^\dagger a_{u_2}^\dagger a_{u_3}a_{u_4}\right\rangle_{\Psi}\right|=\sum_{F(\al)= 0}N^{-1}\rho^{\frac{9}{10}}\sqrt{(\al(v_3)+1)(\al(v_4)+1)}|f(\al)|^2\\\nonumber
+&&\sum_{F(\be)\leq 1}\rho^{\frac{-11}{10}}\sqrt{\frac{\be(v_3)\be(v_4)}{\lambda_{v_3}\lambda_{v_4}}}|f(\be)|^2
\eeqa 
By the  Schwarz inequality, we have that the first term in rhs. is $o(\rho^4)$. Since $v_3, v_4\in P_H$, by \eqref{boundpsifixeduv} we obtain that the second term in rhs. is $o(\rho^{11/4})$. So
\beq
\left|\left\langle a_{v_1}^\dagger a_{v_2}^\dagger a_{v_3}a_{v_4}\right\rangle_{\Psi}\right|\leq \rho^{\frac{11}4}
\eeq
Summing over $v_{j}$'s, we have proved  \eqref{resultasymno24L}.

Finally, we prove \eqref{resultasymno22L2H2} concerning  $v_{1,3}\in P_L$ and $v_{2,4}\in P_H$. Again, with \eqref{no0Fal}, \eqref{no0Fbe} and the bounds on $\lambda$'s in \eqref{lublambdaPL0} and \eqref{ublambdaPH0}, we have 
\beqa\label{temp8.128}
&&\left|\left\langle a_{v_1}^\dagger a_{v_2}^\dagger a_{v_3}a_{v_4}\right\rangle_{\Psi}\right|\leq 
Q_{1}+ Q_{2}+ Q_{3} \\
Q_{1} &= &\sum_{F(\al)= 0}N^{-1}\rho^{\frac{-1}{10}}\sqrt{\frac{\al(v_2)(\al(v_4)+1)}{\left|\lambda_{v_2}\right|}}|f(\al)|^2\\\nonumber
Q_{2}&=&\sum_{F(\be)= 0}N^{-1}\rho^{\frac{-1}{10}}\sqrt{\frac{\al(v_4)(\al(v_2)+1)}{\left|\lambda_{v_4}\right|}}|f(\be)|^2\\
Q_{3} & = & \sum_{F(\al)= 1}\rho^{\frac{-1}{10}}\sqrt{\frac{\al(v_2)(\al(v_4)+1)}{\left|\lambda_{v_2}\right|}}|f(\al)|^2
\eeqa 
By  Theorem \ref{thmA} and the fact $\sqrt{x}\leq x$ for $x\in\N$, we have 
$$
Q_{1}\le  N^{-1}\rho^{\frac{-1}{10}}\lambda_{v_2}^{-1/2}\left(Q_{\Psi}(v_2)+Q_{\Psi}(v_2,v_4)\right)
\le \rho^{3},
$$ 
where we have used the bounds \eqref{ubpsiuroughPH} and \eqref{boundpsifixeduv} on $Q_\Psi(u)$ and $Q_\Psi(u,v)$. Similarly,  we have  $Q_{2}\le \rho^3$.  
Again using the fact $\sqrt{x}\leq x$ for $x\in\N$ , we have 
\beqa\nonumber
Q_{3} & \le & \sum_{F(\al)= 1}\rho^{\frac{-1}{10}}\al(v_2)\left|\lambda_{v_2}\right|^{-1/2}|f(\al)|^2+\rho^{-\frac{1}{10}}\left|\lambda_{v_2}\right|^{-1/2}Q_{\Psi}(v_{2},v_{4}) \\
&\le & \sum_{F(\al)= 1}\rho^{\frac{-1}{10}}\al(v_2)\left|\lambda_{v_2}\right|^{-1/2}|f(\al)|^2+\rho^{3} , \nonumber
\eeqa
where we have used  \eqref{boundpsifixeduv}.
We can estimate the first term in rhs. by  \eqref{temp8.119}. Collecting all these bounds, we have proved that 
\beq
\left|\left\langle a_{v_1}^\dagger a_{v_2}^\dagger a_{v_3}a_{v_4}\right\rangle_{\Psi}\right| \leq \rho^{2.7}
\eeq
Summing over $v_{j}$'s, we have proved  \eqref{resultasymno22L2H2}.
\end{proof}

\section{Interaction  Energy with Four High Momentum Legs I: The Main Term}

We now  estimate of the interaction energy in the  case 4  of Lemma \ref{fourlegslemma1}, i.e., 
$k_i, i=1, 2, 3, 4 $ satisfy  
\beq\label{conditionk1}
k_1+k_2=k_3+k_4,\,\,\, k_1+k_2\neq 0,\,\,\,k_1\neq k_3,\,\,\,k_1\neq k_4, \,\,\,k_i\in P_{H,c}
\eeq
 In the remainder of this paper,  all $p_i$'s, $q_i$'s, $k_i$'s belong to $P_{H,c}$ and $u_i$, $v_i$'s belong to $P_L$.  We start with some special cases.

\begin{lem} Suppose  $k_i$ satisfy  \eqref{conditionk1}. Then  we have
\beqa\label{4kspecialresult1}
&&\sum_{k_1,k_3}\left|  V_{k_{1}-k_{3}}\left\langle a^\dagger_{k_1}a^\dagger_{k_1}a_{k_3}a_{k_4}\right\rangle\right|=o\left(\rho^{5/2}|\Lambda|^2\right)\\\label{4kspecialresult2}
&&\sum_{k_1,k_2}\left|V_{2k_{1}} \left\langle a^\dagger_{k_1}a^\dagger_{k_2}a_{-k_1}a_{k_4}\right\rangle\right|=o\left(\rho^{5/2}|\Lambda|^2\right)
\eeqa
\end{lem}

\begin{proof} By definition of $f$, if $\langle\al| a^\dagger_{k_1}a^\dagger_{k_1}a_{k_3}a_{k_4}|\be\rangle\neq 0$, then
\beq
f(\al)=\sqrt{\left|\frac{\lambda_{v_1}\lambda_{v_2}}{\lambda_{v_3}\lambda_{v_4}}\right|}f(\be)
\eeq
Using Lemma \ref{id4a}, we have 
\beq
\left|\left\langle a^\dagger_{k_1}a^\dagger_{k_2}a_{k_3}a_{k_4}\right\rangle\right|
=\sum_{\be}\sqrt{\left|\frac{\lambda_{k_1}\lambda_{k_2}}{\lambda_{k_3}\lambda_{k_4}}\right|}
\prod_{i=1}^2\sqrt{(\be(k_i)+1)}\prod_{i=3}^4\sqrt{\be(k_i)}|f(\be)|^2
\eeq
Consider first the case $k_1=k_2$ and, by  \eqref{conditionk1},  $k_3\neq k_4$. 
Using the estimates \eqref{ublambdaPH0} for $\lambda_{k_i}$, we have 
\beq
\left|\left\langle a^\dagger_{k_1}a^\dagger_{k_1}a_{k_3}a_{k_4}\right\rangle\right|
=\left|\lambda_{v_3}\lambda_{v_4}\right|^{-\frac12}\rho^{-\frac1{10}}\left(Q_{\Psi}(k_1, k_3, k_4)+Q_{\Psi}( k_3, k_4)\right)
\eeq
Since $\sum_{k_1}Q_{\Psi}(k_1, k_3, k_4)\leq NQ_{\Psi}(k_3, k_4)$, we have 
\beq\nonumber
\sum_{k_1}\left|\left\langle a^\dagger_{k_1}a^\dagger_{k_1}a_{k_3}a_{k_4}\right\rangle\right|
=\left|\lambda_{v_3}\lambda_{v_4}\right|^{-\frac12}\rho^{-\frac1{10}}\left(NQ_{\Psi}( k_3, k_4)+\Lambda k_c^3Q_{\Psi}( k_3, k_4)\right)
\eeq
With $k_3\neq \pm k_4$ and the bound on $Q_{\Psi}(k_3, k_4)$ in \eqref{boundpsifixeduv}, we arrive at the desired result \eqref{4kspecialresult1}.

The case  $k_1=-k_3$ can be proved in a similarly way by using 
 $$\sqrt{(\be(k_1)+1)(\be(k_2)+1)}\le \half(\be(k_2)+\be(k_1)+2).$$
 \end{proof}

By symmetry, we can prove some other special cases such as  $k_1=-k_4$ are negligible. So from now on we focus on the cases
\beq\label{conditionk2}
k_1+k_2=k_3+k_4,\,\,\,k_i\in P_{H,c},\,\,\,k_i\neq\pm k_j\,\,\,{\rm for}\,\,\, i\neq j
\eeq 
This condition will be imposed for the rest of this section. 
Denote by  $M[k_1,k_2]$ the set of all states created by  a soft pair creation $A^{k_1+k_2,\,\,\, k_1/2-k_2/2}$ from another  state, i.e., 
\beq
M(k_1,k_2)\equiv\left\{\be\in M | \exists \al\in M^s_{k_1+k_2} {\rm \; such \; that \;}  \mathcal A^{k_1+k_2,\,\,\, k_1/2-k_2/2}\al=\be\right\}
\eeq
if $k_1+k_2 \in P_L$. Otherwise, we set  $M[k_1,k_2]=\emptyset$.
Notice that $$\left|\mathcal A^{k_1+k_2,\,\,\, k_1/2-k_2/2} \alpha\right\rangle = C  a^\dagger_{k_1}a^\dagger_{k_2}a_{k_1+k_2}a_0 | \alpha\rangle$$
for some normalization constant $C$. Hence for $\beta, \gamma \in M$, if $$\left\langle\be| a^\dagger_{k_1}a^\dagger_{k_2}a_{k_3}a_{k_4}|\gamma\right\rangle\neq 0,$$ 
we have $k_1+k_2=k_3+k_4$ and 
\beq\label{11.1}
\mathcal A^{k_1+k_2,\,\,\, k_1/2-k_2/2}\al=\be\Leftrightarrow  \; \mathcal A^{k_3+k_4,\,\,\, k_3/2-k_4/2}\al=\gamma,   
\eeq

The main contribution of the four nonvanishing leg term is identified in the next lemma. 

\begin{lem}\label{101}
\beqa\nonumber
&&\lim_{k_c,\rho}    \rho^{-5/2}|\Lambda|^{-2}  \sum_{\eqref{conditionk2}}\sum_{\be\in M(k_1,k_2)}
V_{k_1-k_3}f(\be)f(\gamma)\left\langle\be| a^\dagger_{k_1}a^\dagger_{k_2}a_{k_3}a_{k_4}|\gamma\right\rangle
\\\label{resultinMk1k2}
&&\leq 4 \|w^2V\|_1 \frac{1}{3\pi^2}g_0^{3/2}
\eeqa
\end{lem}
\begin{proof}
By  \eqref{11.1}, we have 
\beqa\label{temp8.140}
&&\sum_{\be\in M(k_1,k_2)}f(\be)f(\gamma)\left\langle\be| a^\dagger_{k_1}a^\dagger_{k_2}a_{k_3}a_{k_4}|\gamma\right\rangle\\\nonumber
=&&\prod_{i=1}^4\sqrt{\lambda_{k_i}}\sum_{\al\in M^s_{k_1+k_2}}4|f(\al)|^2|\Lambda|^{-2}\al(0)\al(k_1+k_2)\prod_{i=1}^4\sqrt{(\al(k_i)+1)}
\eeqa
We claim that \eqref{temp8.140} is very close to the following expression:
\beqa\label{temp5.145}
\prod_{i=1}^4\sqrt{\lambda_{k_i}}\sum_{\al\in  M^s_{k_1+k_2} }4|f(\al)|^2|\Lambda|^{-2}\al(0)\al(k_1+k_2)
\eeqa 
For $x_i\geq 0$, we have  
\beq\label{inequalityx1234}
1\leq \sqrt{(x_1+1)(x_2+1)(x_3+1)(x_4+1)}\leq \frac{1}{4}(x_1+x_2+2)(x_3+x_4+2), 
\eeq 
Since $\al(0)\leq N$ and $\al(k_1+k_2)\leq m_c$, we have 
\beq\label{temp8.144}
\frac{\left|\eqref{temp8.140}-\eqref{temp5.145}\right|}{\left|\prod_{i=1}^4\sqrt{\lambda_{k_i}}\right|}\leq \frac{4m_c\rho}{|\Lambda|}\left(\sum_{i}Q_\Psi(k_i)+\sum_{i,j}Q_\Psi(k_i, k_j)\right)
\le \frac{\rho^2}{|\Lambda|}
\eeq
where we have used \eqref{ubpsiuroughPH} and  \eqref{boundpsifixeduv}.

By definition, $Q_\Psi(0,k_1+k_2)=\sum_{\al\in M}\al(0)\al(k_1+k_2)$. Together  with  $\al(0)\leq N$ and $\al(k_1+k_2)\leq m_c$, we have 
\beq
\left|\frac{\eqref{temp5.145}}{\prod_{i=1}^4\sqrt{\lambda_{k_i}}}-4|\Lambda|^{-2}Q_\Psi(0,k_1+k_2)\right|\leq\frac{4m_c\rho}{|\Lambda|}\sum_{\al\in M^a_{k_1+k_2}}|f(\al)|^2
\eeq
Using the bound \eqref{boundMa} concerning $\sum_{\al\in M^a_{k_1+k_2}}$, we have 
\beq\label{temp8.146}
\left|\frac{\eqref{temp5.145}}{\prod_{i=1}^4\sqrt{\lambda_{k_i}}}-4|\Lambda|^{-2}Q_\Psi(0,k_1+k_2)\right|\leq\rho^{3/2}|\Lambda|^{-1}
\eeq
Combining \eqref{temp8.144},  \eqref{temp8.146}, with the bounds on $\lambda$ in \eqref{ublambdaPH0}, we have:
\beq\label{temp8.147}
\left|\eqref{temp8.140}-\prod_{i=1}^4\sqrt{\lambda_{k_i}}4|\Lambda|^{-2}Q_\Psi(0,k_1+k_2)\right|\leq \frac{\rho^{5/4}}{|\Lambda|}
\eeq

Since $\lambda_p=-w_p=-g_pp^{-2}$ for $p\in P_H$ and $|g_p-g_{q}|\leq \const||p|-|q||$, we have for $p, q\in P_{H,c}$ with $p+q\in P_L$ 
\beq
\left|\lambda_{p}-\lambda_{q}\right|\leq \const\eps_H^{-1}\left||p|-|q|\right|\leq \rho^{3/4}
\eeq
This implies $\left||\sqrt{\lambda_{p}}|-|\sqrt{\lambda_{q}}|\right|\leq \rho^{3/8}$. Applying these results 
to  $\prod_{i=1}^4\sqrt{\lambda_{k_i}}$ with  $k_1+k_2=k_3+k_4\in P_L$,  we have
\beq
\left|\prod_{i=1}^4\sqrt{\lambda_{k_i}}-\lambda_{k_1}\lambda_{k_3}\right|\leq \rho^{1/4}
\eeq
Inserting this inequality into \eqref{temp8.147} and using  $Q_\Psi(0,k_1+k_2)\leq Nm_c$, we obtain 
\beq\label{temp8.150}
\left|\eqref{temp8.140}-4\lambda_{k_1}\lambda_{k_3}|\Lambda|^{-2}Q_\Psi(0,v)\right|\leq \rho^{5/4}m_c|\Lambda|^{-1}, \; v=k_1+k_2
\eeq
Summing over $v\in P_L$ and $k_1$, $k_3\in P_{H,c}$, 
we have that the left hand side of \eqref{resultinMk1k2} is equal to 
\beq
\lim_{k_c\to \infty,\rho\to0 }4\|w^2V\|_1 \sum_{v\in P_L}Q_\Psi(0,v) \rho^{-5/2} |\Lambda|^{-2}
\eeq
With \eqref{sum0PL}, we have proved \eqref{resultinMk1k2}. 
\end{proof}

\section{Interaction  Energy with Four High Momentum Legs II:  The Error Terms} 

Our goal in this section is to prove that the interaction energy associated  with four high momentum legs 
which are not covered by Lemma \ref{101}  is negligible. We state it as the following lemma.  Notice that 
Lemma \ref{lemHAS2} follows from the results in the previous two sections and this lemma. 

\begin{lem}\label{113}
\beqa\nonumber
&&\lim_{k_c,\rho}\sum_{\eqref{conditionk2}}\sum_{\be\notin M(k_1,k_2)}\left|\frac{V_{k_1-k_3}}{|\Lambda|}f(\be)f(\gamma)\left\langle\be| a^\dagger_{k_1}a^\dagger_{k_2}a_{k_3}a_{k_4}|\gamma\right\rangle\right|\left(\rho^{5/2}\Lambda\right)^{-1}\\\label{totalresultinMk1k2}=&&0
\eeqa
\end{lem}

We start with the following lemma.
\begin{lem}\label{thehardestlemma}
\beqa\label{resulthardest}
&&\lim_{k_c,\rho}\sum_{\eqref{conditionk2}}\sum_{\be, \gamma: \beta \notin M(k_1,k_2)}\left|f(\be)f(\gamma)\right|  \leq \Lambda\leq o(\rho^{5/2}|\Lambda|^2)
\eeqa
where  the summation is restricted to all $\beta, \gamma\in M$ such that 
\beq\label{defnastk1234}
\left\langle\be| a^\dagger_{k_1}a^\dagger_{k_2}a_{k_3}a_{k_4}|\gamma\right\rangle\neq 0
\eeq
\end{lem}

\begin{proof}  In this section, we use the following notations: 
\beq
\mathcal A_{-k, k}\al\equiv \mathcal A^k \al\rmand \mathcal A_{-k+\frac u2,\,k+\frac u2}\al\equiv\mathcal A^{u,k}\al
\eeq
For any $\{v_1,\cdots, v_t\}\subset P_L$ such that  $v_i\neq \pm v_j, 1\leq i,j\leq t$ and $\al\in M_{v_i}^s, 1\leq i\leq t$,  define 
\beq\label{elementform2}
 M(\al,s,\{v_1,\cdots,v_t\})\equiv \{\prod_{i=1}^{t+s}\mathcal A_{q_{i}, q^{'}_{i}}\al,\,\,\,\,\,q_i,\,\,q_i'\in P_{H,c}, q_{i}+ q^{'}_{i}= u_{i}\}
\eeq
where  $u_i = v_{i}, 1\le i \le t$ and $u_{i} = 0$ otherwise. Since $v_{i}\in P_{L}$ and 
all other momenta are in $P_{H, c}$, $\mathcal A_{q_{i}, q'_{i}}$'s commutes with one another. 

\begin{prop}\label{propchi}
For any  $\chi\in M$, there exists $(\al,s,\{v_1,\cdots,v_t\})$ such that 
\beq\label{resultbegaM}
\chi\in M(\al,s,\{v_1,\cdots,v_t\})
\eeq
\end{prop}

\begin{proof}
By definition of $M$, we can write the state $|\chi\rangle$ as follows:
\beq\label{temp8.164}
|\chi\rangle=\prod_{i=1}^{t}\mathcal A_{p_i,p'_i}  \prod_{k=1}^{s}\mathcal A_{q_k,-q_k} \prod_{j=1}^w(
\mathcal A_{u_j, -u_{j}})^{n_j}|N\rangle, 
\eeq
where $u_j\notin P_{H,c}$, $v_i:= p_i+p'_i \in P_L$, $p_i, p'_i, q_k\in P_{H,c}$. Furthermore, we require that 
$u_j\neq \pm u_{j'}$ for $ j \not= j'$ and $v_i\neq \pm v_{i'}$ for $i \not = i'$. Notice that 
$\mathcal A_{p, p'}$ commute with $\mathcal A_{q, -q}$  so that their orderings are not important. 
Clearly, the choice of 
\beq
\al=\prod_{j=1}^w(\mathcal A_{u_j,-u_j})^{n_j}|N\rangle
\eeq
yields that $\chi\in M(\al, s, \{v_1,v_2\cdots v_t\})$ and this proves the proposition. 
\end{proof}
%

For any $\be, \gamma$ satisfying \eqref{defnastk1234}, we have $\be(u)=\gamma(u)$ for $u\in P_L\cup P_I\cup P_0$. 
{F}rom the proof of  Proposition \ref{propchi}, there exists $(\al, s, \{v_i,1\leq i\leq t\})$ such that 
\beq
\be\rmand \gamma\in M(\al, s, \{v_1,\cdots,v_t\})
\eeq
Notice $\alpha$ is the same for both $\beta$ and $\gamma$ and $\al\in M_u^s$ for any $u\in P_L$. 

For any $(\al,s,\{v_1,\cdots,v_t\})$, define $N(\al,s,\{v_1,\cdots,v_t\})$ as the set of the pairs $(\beta,\gamma)$ such that 
\begin{enumerate}
	\item $\be,\,\,\,\gamma\in M(\al,s,\{v_1,\cdots,v_t\})$
	\item there exist $k_i, i=1, \ldots , 4$ satisfying \eqref{conditionk2}, $\be\notin M(k_1,k_2)$ and \eqref{defnastk1234} holds. 
\item for any other  $\al',s',\{v'_1,\cdots,v'_{t'}\}$ s.t.  $\be,\,\,\,\gamma\in M(\al',s',\{v'_1,\cdots,v'_{t'}\})$, then
\beq\label{stst}
s+t\leq s'+t' 
\eeq
\end{enumerate}
We assume $(\be, \gamma)\in M(\al,s,\{v_1,\cdots,v_t\})$ and \eqref{defnastk1234} holds. Clearly,  $s+t=1$ or $t=0$ implies that  $\be\in M[k_1,k_2]$. 
Hence if $N(\al,s,\{v_1,\cdots,v_t\}) $ is not an empty set then 
\beq\label{stgeq2}
s+t\geq 2 \;  \text { and } \;
t\geq 1
\eeq 

By definition of $N(\al,s,\{v_1,\cdots,v_t\})$, we have
\beqa\label{temp8.173}
&&\sum_{\eqref{conditionk2}}\sum_{\be\notin M(k_1,k_2)}\left|f(\be)f(\gamma)\right|\\\nonumber
\leq &&\sum_{\al, s, \{v_1\cdots v_t\}}\left|N(\al,s,\{v_1,\cdots,v_t\})\right|\max_{\be,\gamma\in M(\al,s,\{v_1,\cdots,v_t\})}{\left|f(\be)f(\gamma)\right|}, 
\eeqa
where $\left|N(\al,s,\{v_1,\cdots,v_t\})\right|$ is the cardinality of $N(\al,s,\{v_1,\cdots,v_t\})$.
By definition of $f$,  if $\be,\gamma\in M(\al,s,\{v_1,\cdots,v_t\})$ then 
\beq
|f(\beta)f(\gamma)|\leq \left|\frac{\al(0)}{|\Lambda|}\right|^{2s+t} \left|\frac {m_c}{|\Lambda|}\right|^{t }\max_{k\in P_H}\{\lambda_k\}^{2t+s}|f(\al)|^2
\eeq
{F}rom \eqref{ublambdaPH0} and $m_c=\rho^{-\eta}$, we have
\beq
\max_{\be,\gamma\in M(\al,s,\{v_1,\cdots,v_t\})}{\left|f(\be)f(\gamma)\right|}\leq (\const\rho^{1-5\eta})^{2s+t}|\Lambda|^{-t}|f(\al)|^2\\\nonumber
\eeq
Together with \eqref{temp8.173}, the right hand side of \eqref{temp8.173} is bounded by 
\beq\label{temp8.176}
\le  \sum_{\al, s, \{v_1\cdots v_t\}}\left|N(\al,s,\{v_1,\cdots,v_t\})\right|(\const\rho^{1-5\eta})^{2s+t}|\Lambda|^{-t}|f(\al)|^2
\eeq
Define $N(\al,s,t)$ and $ N(s,t)$ by 
\beq
 N(\al, s,t)\equiv \max_{ \{v_1,\cdots,v_t\}}\left\{\left| N(\al, s,\{v_1,\cdots,v_t\})\right|\right\}
\eeq
\beq
 N(s,t)\equiv \max_{\al}\left\{ N(\al, s,t)\right\}
\eeq
With \eqref{temp8.176}, we can bound \eqref{temp8.173} by 
\beqa\nonumber
\eqref{temp8.173}&&\leq \sum_{\al, s,t}|f(\al)|^2\sum_{\{v_1\cdots v_t\}} N(\al, s,t)(\const\rho^{1-5\eta})^{2s+t}|\Lambda|^{-t}\\
&&\leq \sum_{s,t}\sum_{\{v_1\cdots v_t\}} N( s,t)(\const\rho^{1-5\eta})^{2s+t}|\Lambda|^{-t}
\eeqa
For fixed $t$ the total number of set $\{v_1\cdots v_t, v_i \in P_L\}$ is bounded by
$$
\sum_{\{v_1\cdots v_t\}}1\leq (\Lambda\rho^{3/2}\eta_L^{-3})^t(t!)^{-1}\leq (\rho^{1-5\eta})^{\frac{3t}2}|\Lambda|^{t}(t!)^{-1}
$$
{F}rom  $t\leq (\Lambda\rho^{3/2}\eta_L^{-3})\leq \rho^{-1.65}$ and  \eqref{stgeq2}, we have 
\beq\label{temp8.178}
\sum_{\eqref{conditionk2}}\sum_{\be\notin M(k_1,k_2)}\left|f(\be)f(\gamma)\right | \leq \sum_{t=1}^{\rho^{-1.65}}\sum_{s: s+t\geq 2} N( s,t)(\const\rho^{1-5\eta})^{2s+\frac{5t}2}(t!)^{-1}
\eeq

\begin{lem}\label{lem11.1}
For any $N(\al, s,\{v_1,\cdots,v_t\})$, $s+t\geq 2$ and $t\geq 1$, we have 
\beq\label{0.025}
\left|N(\al, s,\{v_1,\cdots,v_t\})\right|\leq t\,!\,t^{(\frac {t}2)}|\Lambda|^{\frac{s+t}2+1} (\rho^{-5\eta})^{t+s}
\eeq
\end{lem}

{F}rom this Lemma and  $\Lambda=\rho^{-\frac{25}{8}}$, the right hand side of \eqref{temp8.178} is bounded above by 
\beqa\nonumber
&&\sum_{t\geq 1}^{\rho^{-1.65}}\sum_{s: s+ t>1}\,\left(\rho^{5/2}|\Lambda|^{1/2}t^{1/2}\right)^t\left(\rho^{2}|\Lambda|^{1/2}\right)^s \left(\const\rho^{-35\eta/2}\right)^{t+s}\Lambda\\\nonumber
=&&\sum_{t\geq 1}^{\rho^{-1.65}}(\const\rho^{0.85}t^{1/2})^{t}\sum_{s: s+ t>1}(\const\rho^{0.35})^s\Lambda\leq \Lambda
\eeqa
This proves Lemma \ref{thehardestlemma}.  

 \end{proof}

We now prove Lemma \ref{lem11.1}.
\begin{proof}
Since  $(\beta,\gamma)\in N(\al, s, \{v_1,\cdots, v_t\})$, we can express them as  
\beq\label{112}
\be =\prod_{j=t+1}^{s+t}\mathcal A_{ q_{2j-1},\, q_{2j}}\prod_{i=1}^{t}A_{ q_{2i-1},\, q_{2i}}\al,  \; \gamma=\prod_{j=t+1}^{s+t}\mathcal A_{\widetilde q_{2j-1},\,\widetilde q_{2j}}\prod_{i=1}^{t}A_{\widetilde q_{2i-1},\,\widetilde q_{2i}}\al
\eeq
and $ q_{2i-1}+ q_{2i} = v_i = \widetilde q_{2i-1} + \widetilde q_{2i}$ for $i=1, \ldots, t$, $q_{2j-1}+ q_{2j}= \widetilde q_{2j-1} + \widetilde q_{2j}=0$ for $t+1\leq j\leq s+t$. 
{F}rom \eqref{defnastk1234}, we have 
 \beq
 \{q_1,\cdots,\,q_{2s+2t}\}-\{k_1,k_2\}=\{\widetilde q_1,\cdots,\,\widetilde q_{2s+2t}\}-\{k_3,k_4\}
 \eeq

Denote the common elements in $\{q_i\}$ and $\{\widetilde q_i\}$ by  $p_1$, $p_2,$ $\cdots$, $p_{2s+2t-2}$. Then 
we have  \beq
 \{q_i\}=k_1,\,\,\,k_2,\,\,\,p_1,\,\,\,p_2,\,\,\,\cdots, \,\,\,p_{2s+2t-2},\,\,\,
 \eeq
 \beq
 \{\widetilde q_i\}=k_3,\,\,\,k_4,\,\,\,p_1,\,\,\,p_2,\,\,\,\cdots, \,\,\,p_{2s+2t-2},\,\,\,
 \eeq
We now construct a graph with vertices  $ \{ k_1, k_2, k_3, k_4, p_i, 1\le i \le 2s+2t -2 \} $.  The edges of the graphs consisting of 
$\beta$ edges $(q_{2i-1},q_{2i}), 1\le i \le s+t$ and $\gamma$ edges $(\widetilde q_{2j-1},\widetilde q_{2j}), 1\le i \le s+t$.  
{F}rom  \eqref{defnastk1234},  the graph can be decomposed into two chains and loops. Thus there exist
 $l$, $m_i\in \Z$ and $0<m_1<m_2<......<m_l=s+t$ such that 
%
\beqa\label{temp4.178}
&&k_1\longleftrightarrow p_1\longleftrightarrow p_2\longleftrightarrow p_3\cdots p_{2m_1-1}\longleftrightarrow k_2(\rmor\cdots\,k_4)\\\nonumber
&&k_3\longleftrightarrow p_{2m_1}\longleftrightarrow p_{2m_1+1}\cdots p_{2m_2-2}\longleftrightarrow k_4(\rmor\cdots k_2)
\\\nonumber
&&p_{2m_2-1}\longleftrightarrow p_{2m_2}{\longleftrightarrow}p_{2m_2+1}\cdots p_{2(m_3)-2}\longleftrightarrow p_{2m_2-1}
\\\nonumber
&&\cdots\\\nonumber
&&\cdots
\\\nonumber
&&p_{2m_{l-1}-1}\longleftrightarrow p_{2m_{l-1}}\longleftrightarrow p_{2m_{l-1}+1}\cdots p_{2(m_l)-2}\longleftrightarrow   p_{2m_{l-1}-1}
\eeqa
Here we have relabeled the indices of $p$ and do not distinguish $\beta$ edges and $\alpha$ edges. We also disregard the obvious 
symmetry $k_1 \to k_2$ and $k_3 \to k_4$.  
Due to the condition  \eqref{stst}, the length of the loop must be $4$ or more, i.e., for $3\leq i\leq l$
\beq
m_{i-1}+2\leq m_i
\eeq
Together with $m_l=s+t$, we obtain
\beq\label{uboundl}
l\leq (s+t)/2+1, \quad t \ge 1.
\eeq
Without loss of generality, we assume  for $3\leq i<j\leq l$
\beq
m_i-m_{i-1}\leq m_j-m_{j-1}
\eeq
Denote by $N(\al, s, \{v_1,\cdots,v_t\}, l, \{m_1,\cdots,m_l\})$ the set of all pairs  $(\beta,\,\,\,\gamma)$ having the graph above
and we now estimate its cardinality.

We can add the information between $k_i$'s and $p_i$'s as follows
\beqa
&&k_1\stackrel{w_1}{\longleftrightarrow} p_1\stackrel{\widetilde w_1}{\longleftrightarrow} p_2\stackrel{w_2}{\longleftrightarrow} p_3\cdots p_{2m_1-1}\stackrel {w_{m_1}}{\longleftrightarrow} k_4(\rmor \cdots k_2)\\\nonumber
&&k_3\stackrel{\widetilde w_{m_1}}{\longleftrightarrow} p_{2m_1}\stackrel{w_{m_1+1}}{\longleftrightarrow} p_{2m_1+1}\cdots p_{2m_2-2}\stackrel{{\widetilde w_{m_2}}}{\longleftrightarrow} k_2(\rmor \cdots k_4)
\\\nonumber
&&p_{2m_2-1}\stackrel{w_{m_2+1}}{\longleftrightarrow} p_{2m_2}{\stackrel{\widetilde w_{m_2+1}}{\longleftrightarrow}}p_{2m_2+1}\cdots p_{2(m_3)-2}\stackrel{\widetilde w_{m_3}}{\longleftrightarrow} p_{2m_2-1}
\\\nonumber
&&\cdots\\\nonumber
&&\cdots
\\\nonumber
&&p_{2m_{l-1}-1}\stackrel{w_{m_{l-1}+1}}{\longleftrightarrow} p_{2m_{l-1}}\stackrel{\widetilde w_{m_{l-1}+1}}{\longleftrightarrow} p_{2m_{l-1}+1}\cdots p_{2(m_l)-2}\stackrel{\widetilde w_{m_{l}}}{\longleftrightarrow}   p_{2m_{l-1}-1}\, , 
\eeqa
where  $A\stackrel{c}{\longleftrightarrow} B$ if and only if  $A+B=c$. And $w_i$'s the union of $s$ zero's and $\{v_1,\cdots, v_t\}$, so are $\widetilde w$'s. 
By \eqref{112}, $\beta$ and $\gamma$ is uniquely determined by  $w_i$'s, $\widetilde w_i$'s and one $k_i$ or $p_i$ for  each loop or chain. 
\par 

To bound $| N(\al, s, \{v_1,\cdots,v_t\}, l, \{m_1,\cdots,m_l\}) |$, we note that the sum of momentum in each loop is zero.  
Thus we can count the number of graphs  as follows. 
\begin{enumerate}
	\item  choose the positions of zeros in $\beta$ edges. The total number of choices is less than $2^{t+s}$. 
	\item   choose the positions of $v_1\cdots v_t$ in $\beta$ edges. The total number of choices is  $t!$.
	\item   choose the positions of zeros in $\gamma$ edges. The total number of choices is less than $2^{t+s}$. 
	\item   choose the positions of $v_1\cdots v_t$ in $\gamma$ edges.  We call a loop trivial if all 
	the momenta associated with $\gamma$  edges are zero.   The number of trivial loops is at most $s/2$ since 
	there are at least two $\gamma$ edges per loop. Hence the number of non-trivial loops is at least $l-s/2$. 
	Thus we only have to fix  $v$ in at most $t-(l-s/2)$ edges and  the number of choices is at most $t^{t-l+s/2}$.
\end{enumerate}
Thus we obtain 
\beqa\label{alstvlm}
 &&|N(\al, s, \{v_1,\cdots,v_t\}, l, \{m_1,\cdots,m_l\})|\\\nonumber
  \leq&& (\const)^{t+s}t! t^{(t+s/2-l)}\left(k_c^3\Lambda\right)^{l}\\\nonumber
 \leq &&(\const)^{t+s}t! t^{(t/2)}\left(k_c^3\Lambda\right)^{t/2+s/2+1}
\eeqa
where we have used (\ref{uboundl})
Since  $$| N(\al, s, \{v_1,\cdots,v_t\})|=\sum_{l}\sum_{\{m_1,\cdots,m_l\}} | N(\al, s, \{v_1,\cdots,v_t\}, l, \{m_1,\cdots,m_l\}) | 
$$ and 
\beq
\sum_{l}\sum_{\{m_1,\cdots,m_l\}} 1\leq \const^{s+t}
\eeq 
we have proved \eqref{0.025}.
\end{proof}

We now prove Lemma \ref{113}. 
\begin{proof} Let $\be$, $\gamma\in M$ s.t.$\left\langle\be| a^\dagger_{k_1}a^\dagger_{k_2}a_{k_3}a_{k_4}|\gamma\right\rangle\neq 0$. Using Lemma \ref{id4a} and the definition of $f$, we have 
\beqa\nonumber
&&\left|f(\be)f(\gamma)\left\langle\be|  a^\dagger_{k_1}a^\dagger_{k_2}a_{k_3}a_{k_4}|\gamma\right\rangle\right|=f(\be)f(\gamma)\sqrt{\be(k_1)\be(k_2)\gamma(k_3)\gamma(k_4)}\\\label{temp8.162}
\leq &&|f(\be)|^2\sqrt{\left|\frac{\lambda_{k_3}\lambda_{k_4}}{\lambda_{k_1}\lambda_{k_2}}\right|}(\be(k_1)+\be(k_2))\sqrt{\gamma(k_3)\gamma(k_4)}
\eeqa
{F}rom  the bound on  $\lambda_{k_i}$'s in \eqref{ublambdaPH0} and  $N=\rho^{-17/8}$, we have
\beqa\nonumber
&&\left|f(\be)f(\gamma)\left\langle\be|  a^\dagger_{k_1}a^\dagger_{k_2}a_{k_3}a_{k_4}|\gamma\right\rangle\right|\leq |f(\be)|^2(\lambda_{k_1}\lambda_{k_2})^{-\frac{1}{2}}(\be(k_1)+\be(k_2))\rho^{-\frac{9}{4}}
\eeqa
Since $  Q_\Psi(\{k,m\})$ decays exponentially with $m$ for $k\in P_H$ \eqref{temp5.19}, we have 
\beq
\sum_{\eqref{conditionk2}}\sum_{\be(k_1)>3 \rmor \be(k_2)>3}\left|f(\be)f(\gamma)\left\langle\be| a^\dagger_{k_1}a^\dagger_{k_2}a_{k_3}a_{k_4}|\gamma\right\rangle\right|\leq o(\rho^{5/2}|\Lambda|^2)
\eeq
By symmetry, we have
\beq
\sum_{\eqref{conditionk2}}\sum_{\gamma(k_3)>3 \rmor \gamma(k_4)>3}\left|f(\be)f(\gamma)\left\langle\be| a^\dagger_{k_1}a^\dagger_{k_2}a_{k_3}a_{k_4}|\gamma\right\rangle\right|\leq o(\rho^{5/2}|\Lambda|^2)
\eeq
To prove \eqref{totalresultinMk1k2}, we only have  to focus on the case  $\be(k_{i})\leq 3, i=1, 2$ and $\gamma(k_{i})\leq 3, i=3, 4$. In this case, by \eqref{temp8.162},  we have
\beq\left|f(\be)f(\gamma)\left\langle\be|  a^\dagger_{k_1}a^\dagger_{k_2}a_{k_3}a_{k_4}|\gamma\right\rangle\right|\leq \left|\const f(\be)f(\gamma)\right|
\eeq
Using Lemma \ref{thehardestlemma}, we arrive at the desired result \eqref{totalresultinMk1k2}.
\end{proof}

\section{Proof of Lemma \ref{lemma1}}

The proof of Lemma \ref{lemma1} is standard and only a sketch will be given.  
We first  construct an isometry between functions with periodic 
boundary condition in $[0, L]^3$ and functions with Dirichlet boundary condition in $[-\ell,\,\,\,L+\ell]^3$. 
Denote the coordinates of $\x$ by   $\x = (x^{(1)}, x^{(2)}, x^{(3)})$. 
Let  $h(\x)$ supported on $[-\ell,\,\,\,L+\ell]^3$ be the function $h(\x) = q(x^{(1)}) q(x^{(2)}) q(x^{(3)})$ where 
\beq
q(x)=\left\{ 
\begin{array}{ll}
\cos[(x-\ell)\pi/4\ell], & |x|\leq \ell\\
1, & \ell<x<L-\ell\\
\cos[(x-(L-\ell))\pi/4\ell], & |x-L|\leq \ell\\
0, & \text{otherwise}
\end{array}\right.
\eeq
The function  $q(x)$ is symmetric w.r.t $x=L/2$. Due to the property of cosine, for any  function $\phi$ with 
the period   $L$ we have
\beq\label{proh1}
\int_{\x\in[-\ell,\,L+\ell\,]^3}|h \phi(\x)|^2=\int_{\x\in[0,L]^3}|\phi(x)|^2
\eeq
Thus the map $ \phi\longrightarrow h \phi$
is an isometry: $$
L^2_{
{\rm Periodic}}\left([0,L]^3\right)\to L^2_{
{\rm Dirichlet}}\left([-\ell,L+\ell]^3\right).$$

Let $\chi (\x)$ be the characteristic function of the $\ell$-boundary of $[0,L]^3$, i.e., 
$\chi (\x)=1$ if   $|x^{(\alpha)} |\le \ell$ for some $\alpha = 1, 2$ or $3$ where $|x^{(\alpha)} |$ 
is  the distance on the torus. 
Then standard methods yield the following estimate on the kinetic energy of  $h \phi$
\beqa\label{3.51}
&&\int_{\x\in[-\ell,\,L+\ell\,]^3}|\nabla (h  \phi)(\x)|^2\\\nonumber
\leq &&\int_{\x\in[0,\,L]^3}|\nabla \phi(\x)|^2+\const\ell^{-2}\int  \chi (\x) | \phi(\x)|^2
\eeqa

The generalization of this isometry to higher dimensions is straightforward.  Suppose    
$\Psi(\x_1,\cdots,\x_N)$ is a function with period L.  
Then  for any $u \in \R^{3}$, the map 
 \beq\label{defFu}
{\cal F}^u(\Psi): = \Psi (\x_1,\cdots,\x_N) \prod_{i=1}^N h(\x_i+u)
\eeq
is an  isometry from $L^2_{
{\rm Periodic}}\left([0,L]^{3N}\right)$ to $ L^2_{
{\rm Dirichlet}}\left([-\ell-u,L+\ell-u]^{3N}\right)$. 
Clearly,   ${\cal F}^u$ has the property \eqref{3.51}.

The potential $V$ can be extended to be periodic by defining   $V^P(x-y) = V( [x-y]_{P})$ 
where $[x-y]_{P}$ is the difference of $x$ and $y$ as elements on the torus $[0,L]$. 
%
Since $V$ is nonnegative and has fast decay in the position space,    we have  $V(x-y) \leq  V^P(x-y)$.
{F}rom the definition of    ${\cal F}^u$, we conclude that 
\beq\nonumber
\int |{\cal F}^u(\Psi)|^2V(\x_1-\x_2)\prod_{i=1}^N d\x_i \le 
\int_{[0,\,L]^{3N}} |\Psi|^2 V^P(\x_1-\x_2)\prod_{i=1}^Nd\x_i
\eeq

Therefore,  the energy of two boundary conditions are related by 
\beq
\left\langle H_N\right\rangle_{F^u(\Psi)}\leq \left\langle H_N\right\rangle_{\Psi}+\const\ell^{-2}\sum_{i=1}^N\left\langle \chi (\x_i+u)\right\rangle_{\Psi}
\eeq
Averaging over $u\in[0,L]^3$, we have 
\beq
\int_{[0,L]^3}\left\langle H_N\right\rangle_{F^u(\Psi)}du
\leq L^3\left\langle H_N\right\rangle_{\Psi}+\const\ell^{-1}L^2N
\eeq
So for any $\Psi$ there exists an $u$ such that 
\beq
\left\langle H_N\right\rangle_{F^u(\Psi)}\leq \left\langle H_N\right\rangle_{\Psi}+\const N(\frac{1}{\ell L})
\eeq
If we choose  $\ell$ and $L$ as 
\beq\label{temp2.2}
\ell=\rho^{-25/48},\,\,\,  L=\rho^{-25/24},  
\eeq
 the error term is negligible to the accuracy we need in proving  Lemma \ref{lemma1}. This concludes the proof of Lemma \ref{lemma1}.


\end{document}